\documentclass[a4paper,10pt]{article}
\usepackage[english]{babel}
\usepackage{graphicx,color}
\usepackage{xspace}
\usepackage{hyperref}
\usepackage{url}
\usepackage{amsmath}
\usepackage{amsfonts}
\usepackage{amsthm}
\usepackage{amssymb}
\usepackage{wasysym}
\usepackage[small,nohug,heads=vee]{diagrams}
\diagramstyle[labelstyle=\scriptstyle]

\newtheorem{definition}{Definition}

\newtheorem{theorem}{Theorem}[section]
\newtheorem{lemma}[theorem]{Lemma}
\newtheorem{proposition}[theorem]{Proposition}
\newtheorem{corollary}[theorem]{Corollary}

\newcommand{\field}[1]{\mathbb{#1}}

\newcommand{\mc}[1]{\mathcal{#1}}
\newcommand{\ul}[1]{\underline{#1}}
\newcommand{\cat}[1]{\mathbf{#1}}
\newcommand{\Set}[0]{\cat{Set}}
\newcommand{\colim}[0]{{\rm colim \; }}

\title{Bohrification of local nets of observables}
\author{Joost Nuiten} 

\begin{document}

\maketitle

\begin{abstract}
Recent results by Spitters et. al. suggest that quantum phase space can usefully be regarded as a ringed topos via a process called Bohrification. They show that quantum kinematics can then be interpreted as classical kinematics, internal to this ringed topos.

We extend these ideas from quantum mechanics to algebraic quantum field theory: from a net of observables we construct a presheaf of quantum phase spaces. We can then naturally express the causal locality of the net as a descent condition on the corresponding presheaf of ringed toposes: we show that the net of observables is local, precisely when the presheaf of ringed toposes satisfies descent by a local geometric morphism.
\end{abstract}

\tableofcontents

\section{Introduction}
At the beginning of last century, quantum physics has introduced a completely new way of describing physical systems that differs significantly from the approach people took in classical physics. Instead of viewing observables as functions on some manifold, observables were described as operators acting on a Hilbert space of states. On one hand, this shift in perspective has provided us with a description of an overwhelming amount of new physics, but on the other hand, the new theory turned out to be hard to interpret and seemed `strange' in many ways. To get a grasp on the underlying structure of quantum mechanics, some mathematicians tried to formulate a logical framework in which quantum mechanics could be neatly formulated (cf. \cite{bn36}). Others were looking for certain parallels between classical and quantum physics, trying to put both theories on a common footing.

A decade ago, it was suggested by Butterfield and Isham \cite{bi98} that topos theory could provide a better framework for formulating something like quantum logic. In particular, they noticed that while a quantum phase space does not exist as an ordinary topological space, it does exist as a topos with an internal `space', or locale. This perspective allowed them to give a geometric formulation of the Kochen-Specker theorem, which then precisely stated that this internal phase space had no global points. 

Inspired by this, Spitters et al. \cite{hls09} described a quantum phase space as a topos with an internal locale, using a procedure they called Bohrification. Equivalently, one may view this phase space as a ringed topos, or even a topos with an internal C*-algebra. This description of a phase space fits in the modern approach to geometry, which traces back to Grothendieck who noticed that a geometric space could be formalized by a (locally) ringed topos. Presently, the most advanced theory of general geometry (as discussed in \cite{lu09}) all revolves around regarding locally ringed (higher) toposes as generalized spaces. 

The identification of a quantum phase space with a ringed topos allows the authors of \cite{hls09} to give a clear topos-theoretic picture of quantum kinematics: they show that the quantum kinematics of a system is precisely classical kinematics, internal to this ringed topos. This statement is made precise in section \ref{sec:kin}, which summarizes their results. 

Here we will study the generalization of these constructions from quantum
mechanics to quantum field theory and from plain quantum kinematics to
quantum dynamics: using the formalism of AQFT to describe a quantum field theory on a spacetime $X$
in terms of a copresheaf of algebras of observables, we construct a presheaf of ringed toposes on the opens of $X$, called a Bohrified net on the spacetime $X$.  

Our main result is then a characterization of causal locality of quantum
field theory in these terms: we show that this Bohrified net of
toposes on spacetime satisfies -- over any spatial hyperslice --
descent by local geometric surjections precisely if it comes from a
causally local net. As a direct application we can characterize the
locality of field theories with boundaries by the local descent of the
presheaves of toposes of their chiral boundary field theories.

This result might add to the insights by Isham-D\"oring \cite{id07} and Spitters et. al. that many aspects of quantum theory have a natural formulation when one models quantum phase space as a ringed topos. In particular, it suggests that the ideas by these two groups could be successful when applied to quantum field theory.

\medskip
\noindent{\bf Acknowledgment}\\
\noindent 
I am extremely grateful to Urs Schreiber for all his support during the last months. His enthusiasm and his continuous flow of ideas have been the driving force behind this project. Only the excessive amount of time he took to talk to me and to review my drafts, made this thesis possible. 

I am grateful to Igor Khavkine for useful discussions about the foundations of quantum field theory. I thank Bas Spitters for explaining the subtleties of his work on Bohrification. Finally, I thank Benno van den Berg and Sander Wolters for their comments on an earlier version of this text.

The key idea of this thesis, that causal locality may be encoded
by a sheaf-like property of Bohrified nets, originated in the discussion Urs
had with Bas Spitters during our short visit in Nijmegen.

\section{Bohrification of a C*-algebra}\label{sec:bohr}
The process of Bohrification gives a way to assign a space to a noncommutative C*-algebra (we assume all C*-algebras to be unital). It is inspired by the observation that a lot of information about a C*-algebra is contained in its commutative subalgebras. Indeed, in two major theorems about the structure of quantum mechanics, commutative subalgebras play a main role. The first theorem, by Kochen and Specker, states that there is no way to assign numbers to observables in such a way that it preserves all the structure there is on each of the commutative subalgebras.
\begin{definition}
A KS-map $A\rightarrow\field{C}$ is a map from $A$ to the complex numbers with the property that it is a *-homomorphism when restricted to a commutative subalgebra of $A$.
\end{definition}
\begin{proposition}[Kochen-Specker, \cite{ks67}]\label{prop:ks}
No KS-maps exist on $B(H)$ if ${\rm dim}(H)>2$.
\end{proposition}
One thus finds that it is already a strong requirement to demand a map $A\rightarrow\field{C}$ to be a *-homomorphism, just on each of the commutative subalgebras of $A$. So apparently the structure of a map on $A$ is, to some extend, contained in its restrictions to the commutative subalgebras. This idea is supported even more by the following theorem by Gleason. 

\begin{definition}
A state on a C*-algebra $A$ is a linear functional $\rho: A\rightarrow\field{C}$ that is positive ($\rho(a^*a)\geq 0$) and normalized ($\rho(1)=1$).
\end{definition}

\begin{proposition}[Gleason, \cite{gl57}]
Let $B(H)$ be the C*-algebra of bounded operators on a Hilbert space $H$ with dim$(H)>2$, and let $\rho: B(H)\rightarrow\field{C}$ be a map with the properties that its restriction to each commutative subalgebra is a state and that $\rho(a+ib)=\rho(a)+i\rho(b)$ for all self-adjoint elements $a,b$. Then $\rho$ is in fact a state on the whole of $B(H)$.
\end{proposition}

A state is therefore uniquely defined by its restrictions to each of the commutative subalgebras of $B(H)$. We can then equally well define a quantum state as follows:

\begin{definition}
A quasi-linear functional on a C*-algebra $A$ is a map $A\rightarrow\field{C}$ that is linear \emph{only} on each commutative subalgebra of $A$ and satisfies $\rho(a+ib)=\rho(a)+i\rho(b)$ for all self-adjoint $a,b$. We say that $\rho$ is a quantum state if it is positive ($\rho(a^*a)\geq 0$) and unital ($\rho(1)=1$).
\end{definition}

Thus, one finds that states are in fact defined by how they act locally on a commutative subalgebra. On the other hand, if one has locally on each commutative subalgebra a map that preserves too much structure, there is no way to collect all these maps into a single  map on the whole algebra. So apparently it is important what maps look like when they are localized at a commutative subalgebra. This is one of the facts that are suggestive for Bohrification: whenever one is interested in maps that are local in some sense, classes of such maps can be best considered as presheaves on the corresponding local domain. The above theorems suggest that one is interested in presheafs on the poset of commutative subalgebras. Indeed, Bohrification will be a functor that assigns to each C*-algebra a presheaf topos over the commutative subalgebras, endowed with an internal C*-algebra.

The fact that one is working over the commutative subalgebras of a C*-algebra finds an interesting analogue in the philosophical ideas on quantum mechanics by Bohr (cf. \cite{sc73}). Bohr argued that, although a quantum system itself may be highly nonclassical, it should always be probed in a classical context: an experiment on a quantum system should always be described and interpreted in the classical framework one naturally has of the world. One might identify these `classical contexts' of a quantum mechanical system with the commutative subalgebras of the corresponding C*-algebra. Bohr's `doctrine of classical contexts' then says that all physically relevant information of a quantum system should be encoded in the collection of commutative subalgebras of the C*-algebra that models it.

In the rest of this section we will describe the Bohrification process in more detail. We will explain that Bohrification can be defined as the following functor:

\begin{proposition}
Let $\cat{CStar}_{cr}$ be the category of C*-algebras with as arrows the *-homomorphisms that reflect commutativity. Then Bohrification is a functor $B: \cat{CStar}_{cr}^{op}\rightarrow\cat{cCTopos}$ to the category of toposes with internal commutative C*-algebras. 
\end{proposition}
A *-homomorphism $h: A\rightarrow B$ is said to reflect commutativity if for all $x,y\in A$ one has that $[x,y]=0$ if $[h(x),h(y)]=0$. We will precisely define the category $\cat{cCTopos}$ is section \ref{sec:psh}, but for now just mention that it is a subcategory of the category of ringed toposes. We first assign to each C*-algebra its poset of commutative subalgebras and then show that the category of copresheaves over this poset can be naturally endowed with an internal C*-algebra.

\subsection{Poset of subalgebras}
As we mentioned, we will be mostly interested in the commutative subalgebras of some algebra; in particular, one can consider the poset $C(A)$ of commutative subalgebras of $A$, ordered by inclusion. The assignment of this poset to a C*-algebra induces a functor $C: \cat{CStar}\rTo \cat{Pos}$, since one has for any homomorphism $A\rTo^{h} B$ a functor
\begin{diagram}
C(A)	&	\rTo^{C(h)}		&	C(B)
\end{diagram}
defined by
\begin{diagram}
C			&	\rMapsto	&	h(C)		\\
\dInto			&			&	\dInto		\\
D			&	\rMapsto	&	h(D).
\end{diagram}
Indeed, the image of a C*-algebra under a *-homomorphism is a C*-algebra itself (\cite{heu11}, lemma 5), so $C(h)$ is a well-defined functor. 

The main question is how much information is lost when passing from C*-algebras to their posets of commutative subalgebras.
\begin{lemma}\label{lem:ffiso}
\label{prop:pos}
The functor $C: \cat{CStar}\rightarrow \cat{Pos}$ is faithful and reflects all isomorphisms, except the isomorphism $C(0)$ with $0: \field{C}\rightarrow 0$.
\end{lemma}
\begin{proof}
Let $h,k : A\rightarrow B$ be two *-homomorphisms such that $C(h)=C(k)$. A self-adjoint element $a$ of $A$ generates a commutative subalgebra of $A$ and both $h$ and $k$ map this subalgebra to the subalgebra of $B$ generated by $h(a)$. As a consequence, $k(a)=\lambda(a)h(a)$ for some real number $\lambda(a)$. We now have to show that $\lambda(a)=1$ for all self-adjoint elements $a$.

Indeed, first observe that $\lambda(a)=\lambda(b)$ if $a$ is a multiple of $b$. On the other hand, we have for all self adjoint elements $a,b$ that
$$
\lambda(a+b)(h(a)+h(b))=k(a+b)=\lambda(a)h(a)+\lambda(b)h(b).
$$
This means that $\lambda(a)=\lambda(b)$ even if $a$ and $b$ are linearly independent. We can thus say that $\lambda(a)$ is actually constant, which means that $\lambda(a)=1$ since $h$ and $k$ both map the identity element in $A$ to the identity element in $B$. We conclude that $h$ and $k$ agree on the self-adjoint elements of $A$.

In turn, the restriction of a *-homomorphism $h: A\rightarrow B$ to the self-adjoint part of $A$ uniquely defines $h$, since any element $a\in A$ can be written as $a=\frac{1}{2} (a_++ i a_-)$, where $a_+=a+a^*$ and $a_-=\frac{1}{i}(a - a^*)$ are both self-adjoint elements of $A$. Therefore $C$ is faithful.

We can use the same kind of argument to show that $C$ reflects isomorphisms, by restricting our attention to normal elements only. We will first show that $C$ reflects all embeddings, except for the isomorphism $C(0)$ with $0: \field{C}\rightarrow 0$, which is obviously not induced by an injective *-homomorphism. Indeed, if $h: \field{C}\rightarrow B$ is a *-homomorphism with $B\neq 0$, then $h$ is always an injection and the functor $C(h)$ is always an embedding. 

We therefore only have to show for $A\neq\field{C}$ that a *-homomorphism $h: A\rightarrow B$  is an injection if $C(h)$ is an embedding. Indeed, if $h$ were not injective, there would be some normal element $a\in A\setminus\field{C}\cdot 1$ that would be send to $0$ by $h$.  But then $C(h)$ would send both $\field{C}\cdot 1\subset A$ and the subalgebra $<a>$ generated by $a$ to the algebra $\field{C}\cdot 1\in C(B)$. This means that $C(h)$ cannot be an embedding if $h$ is not an injection.

Similarly, the fact that $C(h)$ is a surjection implies $h$ to be surjective: for any normal element $b\in B$, there is a subalgebra $C\subset A$ whose image is the subalgebra generated by $b$. In particular, there will be some element $a$ in $A$ such that $h(a)=b$. Combining these two results gives that $C$ reflects all isomorphisms, except for the isomorphism $C(0): C(\field{C})=*\rightarrow C(0)=*$.
\end{proof}
In the proof, we have used that $C(A)$ essentially contains all normal elements of $A$. We can therefore extract the set $N(A)$ of normal elements of $A$ from the poset $C(A)$. Moreover, we know for any two commuting elements in $N(A)$ how to add and multiply them. One can formalize this by introducing the concept of a partial C*-algebra, as done in \cite{vdb10}. There the authors also show that in a general setting, this is as far as we can get: there exist non-isomorphic C*-algebras of which the normal parts, with addition and multiplication defined only for commuting elements, are in fact isomorphic.

In particular, we cannot say anything about the commutator between two elements, except whether it is zero or not. If we replace the original product on an algebra by its symmetric product (i.e. $A*B=AB+BA$), we might retrieve some information on this product. Indeed, under certain conditions, two Neumann algebras are Jordan-isomorphic (isomorphic with respect to the symmetric product) if their posets of commutative subalgebras are isomorphic \cite{hdo10}:

\begin{proposition}
Let $A, B$ be two von Neumann algebras without type-$I_2$ summands and $H: C(A)\rightarrow C(B)$ an isomorphism. Then there is a Jordan-isomorphism $h: A\rightarrow B$ that induces the isomorphism $H$ between their posets of commutative subalgebras.
\end{proposition}
\noindent
On the other hand, there exist C*-algebras that are not isomorphic, although the corresponding Jordan algebras are isomorphic (an example is given in \cite{con75}). As a consequence, we cannot lift this result to C*-algebras and *-homomorphisms.

In the light of quantization and reduction, it seems sufficient that only the Jordan structure of an algebra is determined by its poset of commutative subalgebras. Indeed, as noted by for instance Weinstein \cite{bw97}, the symmetric product in a C*-algebra is the quantum analogue of the pointwise multiplication in a Poisson algebra, while the commutator is the quantum analogue of the Poisson bracket. One therefore remembers only the structure that corresponds to a classical algebra of observables without Poisson bracket. The only structure that is lost is the derivation needed to formulate the equations of motions as differential equations: for instance, the Heisenberg equation cannot be formulated without the notion of a commutator. One therefore has to find other ways to formulate dynamics in this setting.

\subsection{Presheaf topos}\label{sec:psh}
Given the poset $C(A)$, one would like to build a topos from it, internal to which one can construct a locale. The easiest way to construct one would be to take $[C(A),\Set]$.  This topos has a tautological object $\ul{A}:C(A)\rightarrow \cat{cCStar}$, defined by $\ul{A}(C)=C$, where $C$ is interpreted as a commutative C*-algebra. Moreover, $\ul{A}$ maps the ordering of the poset to the inclusions of the corresponding subalgebras. As is mentioned in proposition \ref{prop:cstar}, this object $\ul{A}$ is an internal commutative C*-algebra in $[C(A),\Set]$.

Hence to each C*-algebra, we assign a topos with an internal commutative C*-algebra. These toposes form a category:
\begin{definition}
Let $\cat{cCTopos}$ denote the category having
\begin{enumerate}
\item 	as objects (sheaf) toposes with internal commutative C*-algebras.
\item 	as arrows $(\mc{D},A)\rTo (\mc{E},B)$ pairs of arrows
		\begin{diagram}
		\mc{D}	&\pile{\lTo^{f^*} \\ \bot \\ \rTo_{f_*}}	&\mc{E}	\\
				&											&		\\
		f^*B	&\rTo^h 									&	A
		\end{diagram}
		where $f$ is a geometric morphism (with the property that $f^*B$ is again a C*-algebra) and $h$ an internal *-homomorphism in $\mc{D}$.
\end{enumerate}
\end{definition}
\noindent In this way, $\cat{cCTopos}$ is a subcategory of the category of locally ringed toposes that is considered in general geometry.

The direction of the internal *-homomorphism in this definition prevents us from extending the functor $\cat{CStar}\rightarrow\cat{Pos}$ to a functor $\cat{CStar}\rightarrow\cat{cCTopos}$. Indeed, we will have to restrict ourselves to the subcategory $\cat{CStar}_*$ which has as objects C*-algebras and as arrows those homomorphisms $h: A\rightarrow B$ such that $H: C(A)\rightarrow C(B)$ has a right adjoint. This category $\cat{CStar}_*$ is precisely $\cat{CStar}_{cr}$.

\begin{lemma}
A *-homomorphism $A\rTo^h B$ reflects commutativity if and only if the induced functor $C(h): C(A)\rTo C(B)$ has right adjoint.
\end{lemma}
\begin{proof}
If $h$ reflects commutativity, then there is a functor $h^{-1}: C(B)\rTo C(A)$ mapping a commutative subalgebra to its inverse image. It is right adjoint to $C(h)$ because $C\subset h^{-1}(D)$ if and only if $h(C)\subset D$.

Conversely, if $C(h): C(A)\rightarrow C(B)$ has a right adjoint $h^*$, then $h^*(D)$ is the inverse image of $D$ for any $D\in C(B)$. Indeed, we note that $h^{-1}(D)$ is a subalgebra of $A$ (though not necessarily commutative) and that $C(h^{-1}(D))$ has a largest element $h^*(D)$, because any for any commutative $C\subset A$ we have that $h(C)\subset D$ if and only if $C\subset h^*(D)$. In particular does $h^*(D)$ contain all normal elements of $h^{-1}(D)$ (since any normal element in $h^{-1}(D)$ generates a commutative subalgebra which sits inside $h^{-1}(D)$). But if $h^*(D)$ contains all normal elements of $h^{-1}(D)$, then it must contain all of $h^{-1}(D)$. We conclude that $h^*(D)=h^{-1}(D)$, so $h$ reflects commutativity.
\end{proof}
\noindent In particular, we note that embeddings (i.e. monos in $\cat{CStar}$) are arrows in $\cat{CStar}_{cr}$. 

In this way, we find a functor $\cat{CStar}_{cr}^{op}\rightarrow\cat{Pos}$ that assigns to each arrow $h: A\rightarrow B$ in $\cat{CStar}_{cr}$ the inverse image functor $R: C(B)\rightarrow C(A)$ that is right adjoint to the functor $C(h)$. This particular functor can be extended to a functor $\cat{CStar}_{cr}^{op}\rightarrow\cat{cCTopos}$.
\begin{lemma}
There is a functor $\cat{CStar}_{cr}^{op}\rightarrow \cat{cCTopos}$ that sends a C*-algebra to the topos $[C(A),\Set]$ endowed with the internal algebra $\ul{A}: C(A)\rightarrow \cat{Set}$ given by $\ul{A}(C)=C$.
\end{lemma}
\begin{proof}
We wish to use lemma \ref{lem:adj} which applies to categories of presheaves. However, we are interested in categories of copresheaves. This means that a homomorphism $h: A\rightarrow B$ induces an adjoint triple
\begin{diagram}
A		&	& C(A)^{op}		&	& & [C(A),\Set]				\\
\dTo^h	&	& \dTo^{C(h)} \vdash \uTo_R	&	& & \dTo^{C(h)_*} \vdash \uTo^{R_*} \vdash \dTo_{R^*} \\
B		&	& C(B)^{op}		&	& & [C(B),\Set]				
\end{diagram}
in which $R^*$ preserves limits. We pick the geometric morphism $(R_*\vdash R^*)$ as the image of the homomorphism $A\rightarrow B$. 
Observe that $R^*\ul{A}$ is indeed an internal C*-algebra, since $\ul{A}\circ R$ is a presheaf of C*-algebras and therefore an internal C*-algebra in $[C(B),\Set]$.

We should then find an internal *-homomorphism $\mu: R^*\ul{A}\rightarrow \ul{B}$. We can define this natural transformation pointwise by defining $\mu_D: R(D)\rightarrow D$ to be the *-homomorphism $h$ restricted to $R(D)$. Because $C(h)\dashv R$, we see that $h(R(D))=C(h)R(D)\subset D$, so this map is well defined. Furthermore, it clearly defines an internal *-homomorphism.
\end{proof}
We now see why the existence of a left adjoint is necessary for our construction: if we picked the left geometric morphism $C(h)_*\vdash C(h)^*$ instead of the right one, we would have to find a *-homomorphism $\nu: C(h)^*\ul{A}\rightarrow \ul{B}$. Pointwise, this would correspond to a *-homomorphism $\nu_D: h(D)\rightarrow D$ which we would not be able to construct in general.

We thus find a Bohrification functor 
$$B: \cat{CStar}_{rc}^{op}\rTo \cat{cCTopos}\rTo \cat{RingTopos}$$
that assigns to each C*-algebra a presheaf topos with an internal ring, or even a commutative C*-algebra. In the next section, we will show that this presheaf topos is in fact equivalent to a sheaf topos over a topological space. This allows us to describe Bohrification as a functor to the category of ringed spaces $\cat{RingSp}$, which is a subcategory of the category $\cat{RingTopos}$ of ringed toposes.

\subsection{Sheaf topos}
The aim of this paragraph is to establish some properties of the functor $[-,\Set]$ that we will need in section \ref{sec:aqft}. In particular, we show that $[-,\Set]$ actually maps into the topological spaces $\cat{Top}$, after their embedding in $\cat{Topos}$. 

\begin{lemma}\label{lem:locale}
There is a functor $\cat{Pos}\rTo^{\rm Up}\cat{Loc}$ so that there is a natural equivalence
\begin{diagram}[size=0.8em]
\cat{Pos}	&  						& & \rTo^{\rm Up}	& &			& \cat{Loc}\\
			&\rdTo(6,6)_{[-,\Set]}	& &					& &			& \\
			&						& &			& \simeq &			& \\
			& 						& &					& &			& \dTo_{Sh}\\
			&						& &					& &			& \\
			&						& &					& &			& \\
			&						& &					& &			& \cat{Topos}.
\end{diagram}
\end{lemma}
\begin{proof}
In fact, we give a functor $\cat{Pos}\rightarrow\cat{Top}$ by equipping each poset (as a set) with the Alexandrov topology consisting of all its upwards closed subsets. Then each functor (or increasing function) between posets defines a continuous function between the corresponding Alexandrov spaces.

The Alexandrov topology on a poset $P$ has a basis $\mc{B}$ consisting of opens $\uparrow (p):= \{x\in P | x\geq p\} $. Indeed, any upwards closed set $U$ is equal to $\bigvee_{p\in U}\uparrow(p)$. We therefore know that $\cat{Sh}({\rm Up}(P))\simeq \cat{Sh}(\mc{B})$ (see \cite{mml92} for a proof) and we see that $\mc{B}$ has the trivial covering (i.e. the only covering sieve for $\uparrow(p)$ is the maximal sieve), so that the sheaf condition becomes empty. 
Thus, $\cat{Sh}({\rm Up}(P))\simeq [\mc{B}^{op},\Set]$. But it is clear that $\mc{B}^{op}\simeq P$, so for each object $P$ we find an equivalence $\eta_P: \cat{Sh}({\rm Up}(P))\rTo^{\simeq}\left[P,\Set\right]$. Concretely, we find that a copresheaf $G$ in $[P,\Set]$ corresponds to the sheaf $F$ in $\cat{Sh}({\rm Up}(P))$ given by
\begin{equation}\label{eqn:sheaf}
F(U)=\lim_{\uparrow (p) \subset U} G(p) 
\end{equation}

We now have to check that this equivalence is natural in $P$. Let $f: P\rightarrow Q$ be a functor. This induces
\begin{itemize}
\item a geometric morphism
$$
[P,\Set] \pile{\lTo^{f^*} \\ \bot \\ \rTo_{\text{Ran}_f}} [Q,\Set]
$$
\item a frame map $F: {\rm Up}(Q)\rightarrow {\rm Up}(P)$ defined by $F(U)=f^{-1}(U)$, which in turn induces a geometric morphism
$$
\cat{Sh}({\rm Up}(P))\pile{\lTo^{F^*} \\ \bot \\ \rTo_{F_*}} \cat{Sh}({\rm Up}(Q)).
$$
\end{itemize}
We have to show that the following diagram, consisting of all right adjoints, commutes up to isomorphism:
\begin{diagram}[LaTeXeqno]\label{diag:isoiso}
\cat{Sh}({\rm Up}(P))	& \rTo^{F_*}	& \cat{Sh}({\rm Up}(Q))\\
\dTo^{\eta_P}				&				& \dTo_{\eta_Q}\\
[P,\Set]				& \rTo^{f_*}	& [Q,\Set]
\end{diagram}
First note that for a sheaf $E\in \cat{Sh}({\rm Up}(P))$, the induced copresheaf $\eta_PE$ is given by $\eta_P E(p)=E(\uparrow(p))$. For a sheaf $E\in \cat{Sh}({\rm Up}(P))$ we thus find:
$$(\eta_Q F_*E)(q)=E(f^{-1}(\uparrow q))
$$
and
$$
(f_*\eta_P E)(q) = ({\rm Ran}_f \eta_P E)(q) \simeq \lim_{q\rightarrow f(p)} (\eta_PE)(p) \simeq \lim_{p\in f^{-1}(\uparrow(q))} E(\uparrow(p))
.$$
Using the sheaf property of $E$, the last expression is isomorphic to $E(f^{-1}(\uparrow(q)))$. A similar argument now shows that diagram \ref{diag:isoiso} also commutes if we let the functors act on arrows. We then find that diagram \ref{diag:isoiso} commutes up to isomorphism.
\end{proof}
\noindent We thus see that each topos $[C(A),\Set]$ is localic with as underlying locale the Alexandrov space Up$(C(A))$. The proof made essential use of the construction of this Alexandrov space from a poset. We shortly discuss two properties of this construction that will be important in section \ref{sec:aqft}, both of which follow the equivalence between posets and Alexandrov locales discussed in more detail in \cite{car11}. 

First, we note that the functor Up is fully faithful if we take its codomain to be a certain subcategory of $\cat{Loc}$.

\begin{lemma}
The functor $\text{Up}: \cat{Pos}\rightarrow \cat{Loc}$ is fully faithful if we restrict its codomain to the locale morphisms whose inverse images have a left adjoint (in $\cat{Frm}$).
\end{lemma}
\begin{proof}
We observe that for any poset $P$, the products in $\mc{O}($Up$(P))$ are just the set-theoretical intersections. This means that for any functor $f: P\rightarrow Q$ between posets, the frame map Up$(f)=f^{-1}$ preserves all limits and therefore has a left adjoint: for any open $U$ of Up$(P)$, it is given by the smallest upwards closed subset of $Q$ containing the image $f(U)$.

It is obvious that Up is faithful: for two different functors $f,g: P\rightarrow Q$ one finds that $f^{-1}(\uparrow q)\neq g^{-1}(\uparrow q)$ for some element $q\in Q$. 

To show that Up is full, first note that an open $U$ of Up$(P)$ is of the form $\uparrow p$ precisely if it has the property that any cover of $U$ has a subcover consisting of one open (namely $\uparrow p$). We call such an open supercompact \cite{car11}. 
Now, let $F: \text{Up}(Q)\rightarrow\text{Up}(P)$ be a frame morphism and let $L$ be its left adjoint frame morphism. Then $L(\uparrow p)$ is supercompact: indeed, if $L(\uparrow p)\subset \bigvee_i U_i$, then $\uparrow p \subset \bigvee F(U_i)$. But since $\uparrow p$ is supercompact, we can replace this covering by a single open $F(U)$. Then $L(\uparrow)$ is covered by $U$.

In particular, $L(\uparrow p)=\uparrow q$ for a unique element $q\in Q$. We then define the functor $f: P\rightarrow Q$ by sending each $p$ to this corresponding $q$. It is now easily checked that $F$ is given by $f^{-1}$.
\end{proof}

As a direct corollary, we find that $[-,\Set]$ is full and faithful if we let its codomain be the category $\cat{Topos}_{ess}$ of toposes and {\it essential} geometric morphisms.

\begin{lemma}\label{cor:pshfunc}
The functor $[-,\Set]: \cat{Pos}\rightarrow\cat{Topos}_{ess}$ is full and faithful.
\end{lemma}
\begin{proof}
According to the previous lemma, we have to show that under localic reflection, the essential geometric morphisms $\cat{Sh}({\rm Up}(P))\rightarrow \cat{Sh}({\rm Up}(Q))$ correspond precisely to the frame maps $\mc{O}({\rm Up}(Q))\rightarrow \mc{O}({\rm Up}(P))$ with a left adjoint.

Indeed, if the inverse image $f^*: \cat{Sh}({\rm Up}(Q))\rightarrow \cat{Sh}({\rm Up}(P))$ preserves all limits, then so does the localic reflection $f^{-1}: \mc{O}({\rm Up}(Q))\rightarrow \mc{O}({\rm Up}(P))$, since it is given by the restriction of $f^*$ to the subterminals. We conclude that every essential geometric morphism between localic toposes is induced by a frame map with a left adjoint.

Conversely, let $f^{-1}: {\rm Up}(Q)\rightarrow {\rm Up}(P)$ be a frame map with left adjoint $L$. It induces a geometric morphism $(f_*\vdash f^*): \cat{Sh}({\rm Up}(P))\rightarrow \cat{Sh}({\rm Up}(Q))$, where the inverse image functor is given by
\begin{diagram}
\cat{Sh}({\rm Up}(Q))	& \rInto & [\mc{O}({\rm Up}(Q))^{op},\Set] & \rTo^{{\rm Lan}_{f^{-1}}} & [\mc{O}({\rm Up}(Q))^{op},\Set] & \rTo^{\cat{a}} & \cat{Sh}({\rm Up}(P))
\end{diagram}
and $\cat{a}$ is the sheafification functor. We claim that this inverse image functor has a left adjoint, so that $f^{-1}$ induces an essential geometric morphism.

First, the inclusion has the sheafification functor $\cat{a}$ as its left adjoint. By lemma \ref{lem:adj}, ${\rm Lan}_{f^{-1}}$ has a left adjoint given by ${\rm Lan}_{L}$. Finally, we note that the sheafification functor $\cat{a}$ has a left adjoint as well: indeed, let $i: P^{op}\rightarrow \mc{O}({\rm Up}(P))$ be the inclusion sending $p\in P$ to $\{x\geq p \}$. We now note that equation \ref{eqn:sheaf} is precisely the formula for the right Kan extension along $i$, so that Ran$_i: [P,\Set]\rightarrow [\mc{O}({\rm Up}(P))^{op},\Set]$ factors as
\begin{diagram}
[P,\Set] & \rTo^{\simeq} & \cat{Sh}({\rm Up}(P)) & \rInto & [\mc{O}({\rm Up}(P))^{op},\Set].
\end{diagram}
Since ${\rm Ran}_i$ is the direct image of an essential geometric morphism, so is the inclusion $\cat{Sh}({\rm Up}(P)) \rightarrow [\mc{O}({\rm Up}(P))^{op},\Set]$. This proves that the sheafification functor $\cat{a}$ has a left adjoint.
\end{proof}
\noindent We thus find that an essential geometric morphism $[P,\Set]\rightarrow [Q,\Set]$ must be induced by a functor $P\rightarrow Q$.

For our discussion in section \ref{sec:aqft}, it would greatly simplify things if the functor $[-,\Set]: \cat{Pos}\rightarrow\cat{Topos}$ would also preserve limits. However, if we consider its factorization as in lemma \ref{lem:locale}
\begin{diagram}
\cat{Pos} & \rTo^{\rm{Up}} & \cat{Top} & \rTo & \cat{Loc} & \rTo^{\text{Sh}} & \cat{Topos}
\end{diagram}
we note that the middle morphism $\cat{Top}\rightarrow\cat{Loc}$ does not preserve all limits, which means that we cannot guarantee $[-,\Set]$ to preserve all limits if its codomain is $\cat{Topos}$. We \emph{can} conclude that the functor $\cat{Pos}\rightarrow\cat{Top}$ we constructed in the proof of lemma \ref{lem:locale} preserves limits.

\begin{lemma}
The functor Up$: \cat{Pos}\rightarrow\cat{Top}$ from lemma \ref{lem:locale} preserves all limits.
\end{lemma}
\begin{proof}
${\rm Up}$ obviously preserves the terminal object. For products we note that ${\rm Up}(P)\times {\rm Up}(Q)$ has a basis of opens consisting of opens of the form $\{x\geq p\}\times\{y\geq q\}=\{ (x,y)\geq (p,q)\}$ for all $p\in P$, $q\in Q$. These opens form precisely the opens on ${\rm Up}(P\times Q)$. This argument can be applied to arbitrary products as well.

The equalizer $R$ of a diagram $P\pile{\rTo\\ \rTo} Q$ consist precisely of the subposet of $P$ on which both functors agree. An open in $R$ is again an upwards closed set in $R$, which is the same thing as an upwards closed set in $P$, intersected with $R$. Thus Up$(R)$ inherits the subspace topology from $P$, which is precisely the topology on the equalizer $X$ of the diagram $\text{Up}(P)\pile{\rTo\\ \rTo} \text{Up}(Q)$ in $\cat{Top}$. So Up preserves equalizers as well.
\end{proof}

For our discussion in section \ref{sec:aqft}, it will therefore be easier to consider the Bohrification functor $B: \cat{CStar}\rightarrow \cat{RingTopos}$ as a functor into the full subcategory of $\cat{RingTopos}$ consisting of the ringed topological spaces spaces. We then describe Bohrification as a functor $B: \cat{CStar}_{rc}^{op}\rightarrow \cat{RingSp}$ where $\cat{RingSp}$ is the category of topological spaces with sheafs of rings on them.

\section{Quantum kinematics as internal classical kinematics}\label{sec:kin}
The previous section provided a way of assigning a `phase space' to a C*-algebra by means of a locale internal to some sheaf topos. We are now ready to discuss the relevance of this construction in physics and to give a precise formulation of the statement made in the introduction that quantum kinematics is classical kinematics internal to a suitable topos. To do so, we will first give a description of the kinematical structure of classical mechanics.

\subsection{Classical kinematics}\label{sec:class}
Recall that classical systems with finite degrees of freedom are geometrically described by Poisson manifolds. Such a manifold is called the phase space of a system and a point on this manifold describes a state of a specific particle: if we take the example of $T^*Q$ endowed with its canonical symplectic form, a state of a particle is described by its position on the configuration space $Q$ and its momentum. The smooth functions on the manifold are interpreted as the classical observables of the system and form a Poisson algebra.

\begin{definition}
A classical observable is a continuous function from a (Poisson) manifold $M$ to the real numbers.
\end{definition}
\noindent A manifold is determined by its real-valued functions and these functions define an algebra under the pointwise multiplication and addition. In particular, we see that a physical system can be identified with an \emph{algebraic} object.

Since quantum mechanics naturally encorporates a statistical viewpoint, we should extend the concept of a state as a point on a manifold to a probability distribution on a manifold.
\begin{definition}
A classical state is given by a \emph{probability valuation} on the manifold $M$. A probability valuation is a monotone map $\mu: \mc{O}(M)\rightarrow [0,1]$ satisfying
\begin{itemize}
\item $\mu(\emptyset)=0$,
\item $\mu(M)=1$
\item $\mu(U)+\mu(V)=\mu(U\cup V)+\mu(U\cap V)$ and
\item $\mu(\bigvee_i U_i)=\sup_{i\in I}\mu(U_i)$ for any directed family $(U_i)_{i\in I}$ of opens (i.e. a family with the property that for all $U_i$, $U_j$ there is a $U_k$ containing both).
\end{itemize}
\end{definition}
\noindent In particular, every point $x\in M$ defines such a valuation by assigning $U\mapsto 1$ if $x\in U$ and $U\mapsto 0$ otherwise. The states corresponding to a point in $M$ are called the pure states.

We mention that for a compact Hausdorff space $M$ (not necessarily a manifold) there is a one-to-one correspondence between probability valuations on $M$ and probability measures on the Borel subsets of $M$ \cite{hls09}. From a more algebraic perspective, one can define a state as a probability integral on the continuous function of $M$.

\begin{definition}
A probability integral on the self-adjoint part $C_{sa}$ of a commutative C*-algebra $C$ is a linear map $I: C_{sa}\rightarrow\field{R}$ such that $I(1)=1$ and $I(a)\geq 0$ if $a=bb^*$ for some $b$.
\end{definition}

The well-known Riesz-Markov theorem \cite{cs08} states that these two notions of a state coincide.
\begin{proposition}[Riesz-Markov]\label{prop:riesz}
For every compact Hausdorff space $M$ there is a bijection between the probability valuations on $M$ and the probability integrals on $\mc{C}(M)$.
\end{proposition}
\noindent
Thus, a classical state on a compact Hausdorff space $M$ is a probability valuation on $M$ or, equivalently, a probability integral on $\mc{C}(M)$. Thus, this proposition makes clear that a classical state has the algebraic description of a linear map on a commutative C*-algebra. Summarizing, we now find two equivalent ways of describing the kinematics of a classical system:

\begin{itemize}
\item In the \emph{geometric} picture, the observables of a system are the real-valued functions on a space $M$ and the states are valuations on $M$.
\item In the \emph{algebraic} picture, the observables of a system are given by the self-adjoint elements of some commutative C*-algebra $C$ and the classical states are given by probability integrals on $C$.
\end{itemize}
The Riesz-Markov \ref{prop:riesz} and Gelfand duality \ref{sec:gelf} theorems say that these two descriptions are equivalent.

Quantum mechanical kinematics is typically described from an algebraic perspective, which is similar but not exactly the same as the description of classical kinematics. A geometric picture similar to the one in classical physics is however absent, although some constructions like the GNS-construcion do give the quantum theory a somewhat geometrical flavour. The main point of `Bohrification' is that quantum kinematics does fit in any of the above two structures, either algebraic or geometric, if we internalize everything to our topos $[C(A),\Set]$.

\subsection{Quantum kinematics}\label{sec:quant}
We shortly state the well-known formulations of observables and states for a quantum system.

\begin{definition}
A quantum mechanical system is described by a noncommutative C*-algebra. A quantum observable is a self-adjoint element of this C*-algebra.
\end{definition}
From the discussion at the beginning of section \ref{sec:bohr}, we deduce the following definition of a quantum state:
\begin{definition}
A quantum state on a C*-algebra $A$ is a map $A\rightarrow\field{C}$ that is linear \emph{only} on each commutative subalgebra of $A$ and satisfies $\rho(a+ib)=\rho(a)+i\rho(b)$ for all self-adjoint $a,b$. Moreover, $\rho$ should be positive ($\rho(a^*a)\geq 0$) and unital ($\rho(1)=1$).
\end{definition}

With these definitions of quantum states and observables, we can give a precise meaning to the claim that quantum kinematics is internalized classical kinematics. We first notice that all the definitions from section \ref{sec:class} can be internalized to any topos. The following result by Spitters et. al. (\cite{hls09}) then realizes external quantum states as internal classical states:

\begin{lemma}\label{lem:states}
There is a natural bijection between the external states on $A$ and the internal probability integrals on $\ul{A}_{sa}$.
\end{lemma}
To see what the naturality means in this context, we mention that a *-homomorphism $h: A\rightarrow B$ defines
\begin{itemize}
\item a map from the (quasi-)states on $B$ to the states on $A$ by precomposing with $h$.
\item a map from the probability integrals on $\ul{B}_{sa}$ to the integrals on $\ul{A}_{sa}$. Recall from the previous section that $h$ induced a geometric morphism $(R_* \vdash R^*): [C(B),\Set]\rightarrow [C(A),\Set]$ and a *-homomorphism $R^*\ul{A}\rightarrow\ul{B}$. We then precompose an integral with this *-homomorphism to obtain an integral $R^*\ul{A}_{sa}\rightarrow\field{R}$. Applying $R_*$ to this map gives an integral $R_*R^*\ul{A}_{sa}\rightarrow\field{R}$, since $\field{R}$ is in both toposes just the constant presheaf with values in the reals. We finally find a *-homomorphism $\ul{A}\rightarrow R_*R^*\ul{A}$ which we can postcompose with this integral to find our integral $\ul{A}_{sa}\rightarrow\field{R}$.
\end{itemize}
With these two maps acting on the states and integrals, we can check that the given bijection is in fact natural.
\begin{proof}
This follows from the fact that any quasi-state is precisely determined by its action on the commutative subalgebras. In particular, given a quasi-state $\rho: A\rightarrow \field{C}$ one defines a state on $\ul{A}_{sa}$ by $\rho_C=\rho|_{C_sa}: C_{sa}=\ul{A}_{sa}(C)\rightarrow\field{R}(C)=\field{R}$, which clearly defines a positive, unital linear functional on $\ul{A}_{sa}$.

Conversely, a state $\rho: \ul{A}_{sa}\rightarrow\field{C}$ defines a quasi-state $\mu: A\rightarrow\field{C}$ by defining $\mu(x)=\rho_C(x)$ for all $x\in C_{sa}$.  The fact that $\rho$ is a natural transformation makes this assignment consistent. Then for any element $z=x+iy$ with $x,y\in A_{sa}$ we define $\mu(z)=\mu(x)+i\mu(y)$, which clearly gives a quasi-state on $A$.

Suppose we have a *-homomorphism $h: A\rightarrow B$. Then we find that the probability integral corresponding to $\rho\circ h: A\rightarrow\field{C}$ is given by $\rho_C = \rho\circ h|_{C_sa}$. On the other hand, the integral $\rho: \ul{B}_{sa}\rightarrow\field{R}$ is mapped to the map given pointwise by
\begin{diagram}
\ul{A}(C)	&\rInto	& R_*R^*\ul{A}(C)	&\rTo^{h|_{h^{-1}(h(C))}}	& R_*\ul{B}(C)	&\rTo^{\rho|_{h(C)}} &\field{R}(C)
\end{diagram}
or equivalently
\begin{diagram}
C	& \rInto	& h^{-1}(h(C))	& \rTo^h	& h(C)	& \rTo^\rho	& \field{R}.
\end{diagram}

\noindent Clearly this is the same natural transformation as the integral corresponding to $\rho\circ h: A\rightarrow \field{C}$, so the bijection from the theorem is indeed natural in $A$.
\end{proof}

Moreover, we have the observables sitting in the self-adjoint part of the internal C*-algebra $\ul{A}$.
We have therefore retrieved the algebraic picture of classical mechanics, internal to the topos $[C(A),\Set]$: quantum observables are self-adjoint elements of an internal commutative C*-algebra and quantum states are probability integrals on this internal algebra. We now extract from this algebraic description a geometric presentation, using the same methods as for classical physics:
\begin{itemize}
\item We construct a space $\ul{\Sigma}$ such that $\ul{A}\simeq\mc{C}(\ul{\Sigma})$, using the constructive version of the Gelfand duality mentioned in section \ref{sec:gelf}.
\item We identify the probability integrals on $\ul{A}$ with probability valuations on $\ul{\Sigma}$ using a constructive version of the Riesz-Markov theorem \ref{prop:riesz}.
\end{itemize}

\subsection{Geometric description of quantum kinematics}\label{sec:spec}
For each C*-algebra, we found a topos $[C(A),\Set]$ which contained a copresheaf of commutative C*-algebras $\ul{A}$. We can apply the constructive version of Gelfand duality to this internal algebra to obtain an internal locale called the internal spectrum $\ul{\Sigma}_A$ of $\ul{A}$. We will  show that this construction is in some sense functorial, so that the *-homomorphisms we allowed in section \ref{sec:bohr} induce certain maps between the internal locales. 

To do this, we first note that any internal locale in a topos $\cat{Sh}(X)$ can be represented as a bundle over $X$, as shown in \ref{prop:loc}. This bundle can be calculated for the algebras in question, as is done in \cite{hlsw10}, where they obtained the following result:

\begin{proposition}\label{lem:comp}
Let $P$ be a poset and $\ul{A}\in [P,\Set]$ a presheaf of commutative C*-algebras. The spectrum of $\ul{A}$ then corresponds to the bundle $\pi: \Sigma\rightarrow{\rm Up}(P)$ where $\Sigma$ is the topological space with underlying set $\Sigma=\coprod_{x\in P}\Sigma_{\ul{A}(x)}$ and $\Sigma_{\ul{A}(x)}$ is the classical Gelfand spectrum of $\ul{A}(x)$. We say that an subset $\mc{U}\subset\Sigma$ is open if for all $x,y\in P$ the following conditions are satified:
\begin{enumerate}
\item $\mc{U}_x := \mc{U}\cap\Sigma_{\ul{A}(x)}C$ is an open in the Gelfand spectrum $\Sigma_{\ul{A}(x)}$.
\item For $x\leq y$, one finds a map $\ul{A}(x\leq y)^*: \Sigma_{\ul{A}(y)}\rightarrow\Sigma_{\ul{A}(x)}$. Then one should have that $(\ul{A}(x\leq y)^*)^{-1}(\mc{U}_x)\subset\mc{U}_y$.
\end{enumerate}
The map $\pi$ is the canonical projection onto ${\rm Up}(P)$, sending $\Sigma_{\ul{A}(x)}$ to $x$.
\end{proposition}

If we apply this proposition to the tautological functor $\ul{A}\in[C(A),\Set]$, we find a bundle $\Sigma_A\rightarrow{\rm Up}(C(A))$ which gives the external description of our quantum \emph{phase space}. In particular, we thus assign to each C*-algebra a phase space $\ul{\Sigma}_A$, which lives inside a topos instead of just $\Set$. The above proposition thus gives a way to assign to each C*-algebra a space, as we announced in the beginning of section \ref{sec:bohr}.

We should now assign maps of bundles to the arrows in $\cat{CStar}_{cr}$. For a commutativity reflecting *-homomorphism $h: A\rightarrow B$, we found a geometric morphism $R_*\vdash R^*: \cat{Sh}({\rm Up}(C(B)))\rightarrow \cat{Sh}({\rm Up}(C(A)))$; therefore, if we send such a *-homomorphism to a map of bundles, the natural choice of the map between the base spaces would be $R$. Using the notation from section \ref{sec:inloc}, a map $\Sigma_B\rightarrow \Sigma_A$ between the total spaces would now correspond to a map to the pullback $R^\sharp\Sigma_A$.

\begin{diagram}
		&			&			& 		& \Sigma_B\\
		&			&			& \ldTo	\ldTo(1,4)^{\pi_B}\ldTo(4,2)^{h^*}	&		\\
\Sigma_A	&	\lTo	&	R^\sharp\Sigma_A	& 		&		\\
\dTo^{\pi_A}&				&	\dTo		&		&		\\
{\rm Up}(C(A))	&	\lTo^{R}	&	{\rm Up}(C(B))	&		&
\end{diagram}

The map $\Sigma_B\rightarrow R^\sharp\Sigma_A$ is canonically defined from the information from the previous sections. Indeed, we see that
\begin{lemma}
The bundle $R^\sharp\Sigma_A\rightarrow{\rm Up}(C(B))$ is precisely the external spectrum of the C*-algebra $R^*\ul{A}$ in $\cat{Sh}({\rm Up}(C(B)))$.
\end{lemma}
\begin{proof}
Using the expression in proposition \ref{lem:comp} for $\Sigma_A$, we find that $R^\sharp\Sigma_A$ is precisely the set $\coprod_{D\in C(B)}\Sigma_{R(D)}$ with the topology precisely as in proposition \ref{lem:comp} and with the obvious projection map to ${\rm Up}(C(B))$. The computation of the spectrum of $R^*\ul{A}$ as in proposition \ref{lem:comp} will give us precisely the same bundle.
\end{proof}

The discussion in section \ref{sec:psh} gave us an internal *-homomorphism $R^*\ul{A}\rightarrow\ul{B}$, which corresponds by Gelfand duality to a continuous map $\ul{\Sigma}_B\rightarrow\ul{\Sigma}_{R^*A}$. One can check that the external version of this arrow (which is essentially the Gelfand dual of the *-homomorphism $h$ restricted to each commutative subalgebra) precisely gives the map from proposition \ref{lem:comp}. We thus conclude that
\begin{itemize}
\item Bohrification induces a functor $\Sigma: \cat{CStar}_{cr}^{op}\rightarrow\cat{Bund}$ that assigns to each C*-algebra a bundle of locales.
\end{itemize}

On the other hand, we will be mainly interested in the internal locale instead of the bundle. For such a formulation, we have to define a category of toposes with internal locales.

\begin{definition}
Let $\cat{SpTopos}$ be the category having
\begin{itemize}
\item as objects pairs $(\mc{E},L)$ where $\mc{E}$ is a topos and $L$ is a locale in $\mc{E}$
\item as arrows pairs consisting of a geometric morphism $(f_* \vdash f^*): \mc{E}\rightarrow\mc{F}$ and a map of locales $f_*L\rightarrow M$ in $\mc{F}$.
\end{itemize}
\end{definition}
\noindent
Observe that this definition makes sense since we saw that the direct image of a geometric morphism preserved locales. On the other hand, the map $\ul{\Sigma}_B\rightarrow R^\sharp\ul{\Sigma}_A$ corresponds via the adjunction mentioned at the end of section \ref{sec:inloc} to a continuous $R_*\ul{\Sigma}_B\rightarrow \ul{\Sigma}_A$. The previous discussion then makes clear that
\begin{itemize}
\item Bohrification gives a functor $B: \cat{CStar}_{cr}^{op}\rightarrow \cat{SpTopos}$ that assigns to each C*-algebra a phase space, internal to a topos.
\end{itemize}
We have therefore justified the claim we made in section \ref{sec:bohr} that Bohrification gave a way to assign (generalized) spaces to C*-algebras.

We can now again realize quantum kinematics as internal classical kinematics, where this time we give a geometrical description. Indeed, as shown in \cite{hls09}, observables can be realized as real-valued functions on the phase space $\ul{\Sigma}_A$:

\begin{proposition}
There is an injective map $\delta: A_{sa}\rightarrow\mc{C}(\ul{\Sigma}_A, \field{IR})$ with the property that $a\leq b$ iff $\delta(a)\leq\delta(b)$.
\end{proposition}
\noindent
Here $\field{IR}$ is the so called interval domain, which in $\Set$ corresponds to the set $\field{IR}$ of compact intervals $[a,b]$, which has as basic opens the subsets $\{[a,b]\subset (x,y) \}$ for every open interval $(x,y)$. In the presheaf category $[C(A),\Set]$ the interval domain $\field{IR}$ is the constant presheaf assigning to each point the interval domain in $\Set$.

On the other hand, we have already shown in lemma \ref{lem:states} that quantum states correspond to probability integrals on $\ul{A}$. We use a constructive version of the Riesz-Markov theorem \ref{prop:riesz} to identify these probability integrals on $\ul{A}$ with probability valuations on $\ul{\Sigma}_A$:

\begin{lemma}[\cite{cs08}]
The Riesz-Markov theorem \ref{prop:riesz} holds in any topos $\mc{E}$.
\end{lemma}
\noindent In particular, any integral $\ul{A}\rightarrow\field{R}$ gives a valuation $\mc{O}\ul{\Sigma}\rightarrow [0,1]_l$, where $[0,1]_l$ are the lower reals instead of the Dedekind reals: a lower real is described by an open interval in the rationals that is downwards closed. In $\Set$ these are just the open intervals $(-\infty, r)$ for any real number $r$, so in $\Set$ the lower reals and the Dedekind reals coincide. However, this need not hold in an arbitrary topos.

We can say that observables correspond to continuous maps $\Sigma\rightarrow\field{IR}$ and that states correspond to valuations $\mc{O}\ul{\Sigma}\rightarrow [0,1]_l$. Thus, due to the constructive Gelfand and Riesz-Markov theorems, the quantum kinematical structure in $\Set$ corresponds to a geometric picture of classical kinematics in $[C(A),\Set]$:

\begin{itemize}
\item There is a phase space $\ul{\Sigma}$ such that quantum observables give continuous, real-valued functions $\ul{\Sigma}\rightarrow\field{IR}$.
\item Quantum states are then given by valuations $\mc{O}(\ul{\Sigma})\rightarrow [0,1]_l$.
\end{itemize}

This summary shows that Bohrification builds a strong parallel between quantum and classical kinematics. However, the situation is not precisely the same: the mathematical concept of a locale in a topos does not share all the nice properties that a normal topological space has. In particular, there is a lack of `pure states' in the classical sense, because the internal locale does not have any points (that is, locale maps from the point $*$). Indeed, recall from section \ref{sec:class} that in classical physics, a pure state is described by a point of the phase space. The fact that the internal phase space has no global points is now a direct consequence of the no-go theorem by Kochen and Specker (cf. \cite{ks67}) discussed in section \ref{sec:bohr}.

\begin{proposition}
There is a bijective correspondence between KS-maps on $A$ and (internal) points of $\ul{\Sigma}$ in $[C(A),\Set]$.
\end{proposition}

\begin{proof}
Suppose we have a KS-map $v: A_{sa}\rightarrow\field{R}$. This defines a natural transformation $v: \ul{A}_{sa}\rightarrow\field{R}$ by letting $(v_C)=v|_{C_{sa}}: C_{sa}\rightarrow\field{R}$. In particular, this defines a *-homomorphism $\ul{A}\rightarrow\field{C}$ (since $\ul{A}=\ul{A}_{sa}\oplus i\ul{A}$). As is also the case for the classical Gelfand spectrum, the points of the spectrum of $\ul{A}$ are precisely the *-homomorphisms to $\field{C}$. We thus see that the points of the spectrum of $\ul{A}$ correspond precisely to the valuations on $A$.
\end{proof}
\begin{corollary}[Kochen-Specker]
If $A=B(H)$ with dim$(H)>2$, then $\ul{\Sigma}$ has no points.
\end{corollary}
\begin{proof}
This is just a reformulation of the original Kochen-Specker theorem \ref{prop:ks}, using the previous lemma.
\end{proof}

This shows that quantum kinematics looks like classical kinematics, internal to a topos. The only difference remains that the internal phase space contains no points and as a consequence, that quantum mechanics has no pure states in the classical sense.

\section{Bohrification of a local net}\label{sec:aqft}
The previous section summarized how Bohrification identified quantum kinematics with the classical kinematics on some internal phase space. We will now try to add a dynamical flavour to this picture using the framework of algebraic quantum field theory (see for a mathematical account e.g.  \cite{hal06}).

The idea of AQFT is to characterize a quantum field theory on a spacetime $X$ by the assignment of algebras of local observables to each open subset of $X$. These observables present what can be measured by performing an experiment within that certain region of space and time. The fact that within more space and time, certainly more things can be measured, results in saying that this assignment gives a copresheaf $A: \mc{O}(X)\rightarrow \cat{CStar}_{inc}$. Finally, one has to impose some kind of locality condition on this copresheaf, since any relativistic quantum theory should respect the causal structure of spacetime. We will therefore shortly describe the causal structure of spacetime and state what axioms should then be satisfied by a net of observables.

\subsection{Local nets}
The mathematical structure of a spacetime is given by a Lorentzian manifold. Recall that a Lorentzian manifold is a manifold with a metric of signature $(1,p)$, so that at each point the tangent space decomposes in a timelike and spacelike part, separated by the light cone of tangent vectors whose length is $0$. We will keep the $(n+1)$-dimensional Minkowski space $\field{R}^{1+n}$ in mind as a specific example.

We can then say that a curve is, for example, timelike if its tangent vectors are all timelike. Two points $x,y$ in a Lorentzian space $X$ are called \emph{spacelike separated} if there is no timelike or lightlike curve from $x$ to $y$. In Minkowski space this would mean that the line from $x$ to $y$ is spacelike. From a physical point of view, this means that no signals can be send from $x$ to $y$ or the other way around. Two subsets $U,V$ of $X$ are spacelike separated if all $x\in U$, $y\in V$ are spacelike separated.

A connected hypersurface $S$ in a Lorentzian space is said to be a \emph{Cauchy surface} if every timelike or lightlike curve intersects $S$ in precisely one point. This condition formalizes the idea that the points on $S$ give a space at one specific time. In Minkowski space, any plane spanned by only spacelike vectors forms a Cauchy surface.

To formulate a locality condition on the copresheaf $A: \mc{O}(X)\rightarrow \cat{CStar}_{inc}$, we will restrict its domain to a subcategory of opens that nicely fits the causal structure of spacetime sketched above. In particular, we will define it on the set of causal complete opens of a Lorentzian space $X$.
\begin{definition}
For any open $O$ in a Lorentzian space, let $O'$ be its causal complement 
$$
O'=\{ x\in X | x \text{ is spacelike separated from O} \}^{int}.
$$
Let $\mc{V}(X)$ be the set of connected causally complete opens, i.e. opens such that $O''=O$.
\end{definition}
In two dimensions, this set consists of causal diamonds and wedges. Under the inclusions, $\mc{V}(X)$ forms a poset that has all limits and colimits. It then has the natural structure of a site if we say that $U$ is covered by $U_i\rightarrow U$ iff $\bigcup U_i=U$.\\

An algebraic quantum field theory will now be given by a copresheaf $\mc{V}(X)\rightarrow\cat{CStar}_{inc}$ on the category of causally complete opens of a Lorentzian space $X$. We will call such a copresheaf a \emph{net} on $X$. A quantum field theory will now be described by a net that matches the causal structure of $X$, a so-called local net.

\begin{definition}
A local net is a copresheaf $A: \mc{V}(X)\rightarrow \cat{CStar}_{inc}$ with the property that for any two spacelike separated opens $O_1$ and $O_2$ in $X$, the algebras $A_{O_1}$ and $A_{O_2}$ mutually commute in $A_{O_1\vee O_2}$.
\end{definition}
This definition basically says that relativistic independence (i.e. regions being spacelike separated) implies quantum mechanical independence (i.e. that two observables from separated regions commute). All observables are contained in a large C*-algebra $\mc{A}=A_{X}$, so we can assume that all algebras are subalgebras of $\mc{A}$. 

On top of these basic axioms, one often considers local nets that satisfy some additional conditions, like equivariance under the action of the Poincar\'e group if one considers nets on Minkowski space, or the existence of a vacuum sector (see \cite{hal06}). We will impose two conditions on our net that often arise in discussions of AQFT (see for instance \cite{ha96}): we will require our nets to be both \emph{additive} and \emph{strongly local}.

\begin{definition}[\cite{gl06, mu98}]
Let $A:\mc{V}(X)\rightarrow\cat{CStar}_{inc}$ be a net. We say that $A$ is additive if for any two spacelike separated opens $O_1$ and $O_2$ such that $\overline{O_1}\cap\overline{O_2}=\{*\}$, one has that $A(O_1\vee O_2)=A(O_1)\vee A(O_2)$ where the latter is taken in $\mc{A}=A(X)$.
\end{definition}

\begin{definition}
We say that a net $A:\mc{V}(X)\rightarrow\cat{CStar}_{inc}$ is strongly local if it is local and has the property that for any two spacelike separated opens $O_1$ and $O_2$, and any pair of commutative subalgebras $C_1\subset A(O_1)$ and $C_2\subset A(O_2)$, one finds for the algebra $C_1\vee C_2\subset A(O_1\vee O_2)$ generated by them that $(C_1\vee C_2)\cap A(O_1)=C_1$ and $(C_1\vee C_2)\cap A_2=C_2$.
\end{definition}

Being additive in some sense expresses the local character of AQFT: the observables in two small opens $O_1$ and $O_2$ suffice to describe everything that can be observed in the larger open $O_1\vee O_2$ generated by them. For nets on $\field{R}^2$, this condition follows from the so-called \emph{split-property for wedges}, as discussed in \cite{gl06}. 

Strong locality is precisely the kind of locality condition we need in theorem \ref{prop:main}. This condition holds for nets satisfying \emph{Einstein causality}, which is one of the axioms for an AQFT imposed in \cite{bf09}. 
\begin{definition}
A net $A:\mc{V}(X)\rightarrow\cat{CStar}_{inc}$ is called Einstein causal if for any two spacelike separated opens $O_1$ and $O_2$, one has that the inclusions $A(O_1)\subset A(O_1\vee O_2)$ and $A(O_2)\subset A(O_1\vee O_2)$ factor over the tensor product
\begin{diagram}
 & & A(O_1\vee O_2) & &\\
& \ruInto & \uInto & \luInto & \\ 
A(O_1) & \rInto & A(O_1)\otimes A(O_2) & \lInto & A(O_2).
\end{diagram}
\end{definition}
\noindent An Einstein causal net is automatically local, since $A(O_1)$ and $A(O_2)$ commute in their tensor product, so also in $A(O_1\vee O_2)$. Furthermore, an Einstein causal net is indeed stronlgy local: in this case one has for two commutative subalgebras $C_1$ and $C_2$ that $C_1\vee C_2=C_1\otimes C_2$, so that $(C_1\otimes C_2)\cap A_1=C_1$.

In \cite{bf09} it is argued that Einstein causality expresses that the subsystems localized at $O_1$ and $O_2$ are completely independent: ordinary locality only states that it does not matter whether you first do a measurement in $O_1$ and then one in $O_2$, or the other way around. Einstein causality adds to this that the subsystems localized at $O_1$ and $O_2$ are even statistically independent: a state $\rho_1$ on the system at $O_1$ and a state $\rho_2$ on the system at $O_2$ give a product state $\rho_1\otimes\rho_2$ on composite system. Einstein causality thus expresses that there is no way for the subsystem localized at $O_1$ to interact with the subsystem localized at the spacelike separated region $O_2$.\\
\\
\noindent
In the next section, we will need to restrict our local nets to spatial hypersurfaces. If one fixes a Cauchy surface $S\rInto X$, one can consider the restriction of a local net to $S$, which is defined as follows:
\begin{definition}
Let $S\rInto X$ be a Cauchy surface and $A$ a local net. To each connected open $U\subset S$ one can assign the smallest connected causally complete open $O_U\subset X$ that contains $U$. 

We then define the restriction of $A$ to $S$ to be the net $A|_S: \mc{V}(S)\rightarrow\cat{CStar}_{inc}$ that sends each connected open $U\subset S$ to the algebra $A|_S(U)=A(O_U)$.
\end{definition}
\noindent We can indeed construct this net since for $U\subset V\subset S$ one has that $O_U\subset O_V$, so that one finds an inclusion $A(O_U)\rInto A(O_V)$. If we consider a spacelike line $S$ in $\field{R}^{1+1}$, an interval $I\in\mc{V}(S)$ corresponds to the causal diamond in $\field{R}^{1+1}$ that has $I$ as its spacelike axis.

In Minkowski space $\field{R}^{1+n}$, a direct consequence of this construction is that the restriction $A|_S$ of a local net is local in the following sense: for $U$ and $V$ disjoint connected opens of $S$ one has that $A|_S(U)$ and $A|_S(V)$ commute in $A|_S(S)$. Indeed, if $U$ and $V$ are disjoint, then $O_U$ and $O_V$ are spacelike separated.\\

\noindent
We will now turn to the application of the constructions from sections \ref{sec:bohr} and \ref{sec:kin} to these local nets. The Bohrification functor from section \ref{sec:bohr} now enables us to assign to a local net a topos-valued functor. Indeed, let $A:\mc{V}(X)\rightarrow \cat{CStar}_{inc}$ be a local net. Since $\cat{CStar}_{inc}$ is a subcategory of $\cat{CStar}_{cr}$, we can postcompose with the Bohrification functor $\cat{CStar}_{cr}\rightarrow \cat{RingTopos}^{op}$ to  ringed toposes. In fact, as we remarked in section \ref{sec:sh}, we can also consider the Bohrification functor as mapping into the category $\cat{RingSp}$ of ringed spaces, seen as a subcategory of the ringed toposes.

\begin{lemma}
Bohrification induces a functor 
$$[{\mc{V}(X)},\cat{CStar_{inc}}]\rightarrow[{\mc{V}(X)^{op}},\cat{RingSp}],$$
which is faithful and reflects isomorphisms.
\end{lemma}
\begin{proof} We have for any functor $F: \mc{D}\rightarrow\mc{E}$ that postcomposition induces a functor $F_*: [\mc{C},\mc{D}]\rightarrow [C,\mc{E}]$. Moreover, if $F$ is faithful and reflects isomorphisms, then $F_*$ has these properties as well. 

Indeed, let $\mu_1,\mu_2\in[\mc{C},\mc{D}](A,B)$ be two distinct natural transformations. Then certainly there is some $C\in\mc{C}_0$ so that $\mu_1(C)\neq\mu_2(C)$. Since $F$ is faithful, this implies that $(F_*\mu_1)(C)=F(\mu_1(C))$ is distinct from $(F_*\mu_2)(C)$, so that $F_*\mu_1\neq F_*\mu_2$.

Similarly, for $F_*\mu: F_*A\rightarrow F_*B$ an isomorphism, we find that $\mu$ is a natural transformation such that $\mu(C)$ is an isomorphism, since $(F_*\mu)(C)=F(\mu(C))$ is. Thus, $\mu$ is a natural isomorphism.
\end{proof}

We will call the image of a net under this functor its Bohrified net.
\begin{definition}\label{def:bnet}
For any net $A: \mc{V}(X)\rightarrow\cat{CStar}_{inc}$, we define its Bohrified net $B(A): \mc{V}(X)^{op}\rightarrow \cat{RingSp}$ to be the functor that sends an open $O$ to the presheaf category $[C(A(O)),\Set]$ (which is the sheaf topos corresponding to the space Up$(C(A))$) together with the internal ring $\ul{A}(O)$. An inclusion $O_1\rInto O_2$ is sent to the pair
\begin{align}
[C(A(O_2)),\Set] &\pile{\lTo^{R^*} \\ \bot \\ \rTo_{R_*}} [C(A(O_1)),\Set]\nonumber\\
R^*\ul{A}(O_1)&\rTo^{i}\ul{A}(O_2).\nonumber
\end{align}
Here the geometric morphism $R_*\vdash R^*$ is induced by the restriction map $R: C(A(O_2))\rightarrow C(A(O_1))$ that sends a commutative subalgebra $C\subset A(O_2)$ to $C\cap A(O_1)$. The natural transformation $i$ is given by the inclusion
$$
i_C: \ul{A}(O_1)_{R(C)}=C\cap A(O_1)\rInto C= \ul{A}(O_2)_{C}.
$$
\end{definition}
We find for each net $A$ a Bohrified net $B(A): \mc{O}(X)^{op}\rightarrow\cat{RingSp}$. The natural question to ask now is whether this functor satisfies some kind of sheaf property. In the next section we will show that, though it is not a sheaf, this Bohrified net is close to being a sheaf.

\subsection{Locality as local descent}
The aim of this section will be to investigate whether the Bohrified net $B(A): \mc{O}(X)^{op}\rightarrow\cat{RingSp}$ from definition \ref{def:bnet} satisfies some kind of sheaf property. As discussed in appendix \ref{sec:sh}, this is a reasonable question since the sheaf property can be formulated for any functor that maps into a complete category, while $\cat{RingSp}$ is indeed complete according to the remark after lemma \ref{lem:ringtopos}. We will show that the restriction of a local net to a spacelike line will almost be a sheaf; in particular consider a spacelike line and let an interval $K$ on it be covered by two intervals $I$ and $J$. Then for the Bohrified net 
$$
B(A): \mc{O}(X)^{op}\rightarrow\cat{RingSp}\rightarrow\cat{RingTopos}
$$
their category of \emph{matching families} (see section \ref{sec:sh}) is now defined as the pullback in the category of ringed spaces
\begin{diagram}
B(A)(I) \times_{B(A)(I \cap J)}  B(A)(J)	& \rTo & B(A)(J)\\
\dTo										&		& \dTo_{B(A)(I\cap J\rightarrow J)}\\
B(A)(I)					& \rTo_{B(A)(I\cap J\rightarrow I)\qquad}	& B(A)(I\cap J)
\end{diagram}
where the arrow $B(A)(I\cap J\rightarrow I)$ in $\cat{RingSp}$ consists of the geometric morphism $R(I)_*\vdash R(I)^*$ that is induced by the restriction map $R(I): C(A(I))\rightarrow C(A(I\cap J))$ and a ring homomorphism $h: R(I)^*\ul{A}(I\cap J)\rightarrow \ul{A}(I)$. The descent morphism (\ref{sec:sh})
$$
B(A)(K)\rightarrow B(A)(J) \times_{B(A)(I \cap J)}  B(A)(J)
$$
is induced from the restriction maps
$$
B(A)(K)\rightarrow B(A)(I)
$$
from definition \ref{def:bnet}. It consists of
\begin{itemize}
\item a geometric morphism 
\begin{diagram}
[C(A(K)),\Set] &\rTo^{r_*\vdash r^*} & [C(A(I)),\Set]\times_{[C(A(I\cap J)),\Set]} [C(A(J)),\Set]
\end{diagram}
where the pullback is constructed in the subcategory $\cat{Top}$ of $\cat{Topos}$
\item and a ring homomorphism 
$$
r^*\left(\ul{A}(I)\coprod_{\ul{A}(I\cap J)}\ul{A}(J)\right)\rightarrow\ul{A}(K)
$$
where $\ul{A}(I)\coprod_{\ul{A}(I\cap J)}\ul{A}(J)$ is the pushout ring as defined in lemma \ref{lem:ringtopos}.
\end{itemize}

If $B(A)$ were a sheaf, this descent morphism would be an equivalence for all $I$ and $J$. We will show that for a Bohrified \emph{local} net, this descent morphism will not be an equivalence, but the geometric morphism will be a \emph{local geometric surjection}. Recall that a \emph{local} geometric morphism $(f_*\vdash f^*): \mc{E}\rightarrow\mc{F}$ is a geometric morphism such that the direct image $f_*$ has a further right adjoint which is full and faithful (cf. \cite{joh02} section C3.6). Local geometric morphisms model `infinitesimal thickenings': if a sheaf topos $\cat{Sh}(X)$ has a local geometric surjection to the point $\cat{Sh}(*)=\Set$, then $X$ is the `infinitesimal thickening' of a point, in the sense that there is a point whose only neighbourhood is the whole space $X$. This specific example arises if one considers the topos $\cat{Sh}(\text{Spec}(R))$ of sheaves over the spectrum of a local ring.

Specifically, we will prove the following theorem:
\begin{theorem}\label{prop:main}
Let $A: \mc{V}(\field{R}^{1+1})\rightarrow\cat{CStar}_{inc}$ be an additive net. Then $A$ is strongly local, precisely when for any spacelike line and two intervals $(I,J)$ on it, one finds for the descent morphism
\begin{diagram}
B(A)(I\vee J)&	\rTo^r & B(A)(I)\times_{B(A)(I\cap J)}B(A)(J).
\end{diagram}
that
\begin{itemize}
\item the descent morphism of the toposes is a local geometric morphism 
\begin{diagram}
[C(A(I \cup J)),\cat{Set}]&	\pile{\lInto^{r^*}\\ \bot \\ \rTo_{r_*}\\ \bot \lInto} & [C(A(I)),\Set]\times_{[C(A(I\cap J)),\Set]}[C(A(J)),\cat{Set}]
\end{diagram}
so in particular a geometric surjection and
\item that the descent morphism in $\cat{Ring}^{op}([C(A(I\cup J)),\cat{Set}])$ is an epi.
\end{itemize}
\end{theorem}
Furthermore we notice that if $A$ is just local, we do find an extra right adjoint to the geometric morphism $r_*\vdash r^*$, which is however not full and faithful.

We will prove this theorem in three steps: first we consider the descent morphism induced on the posets of commutative subalgebras, to conclude in section \ref{sec:topdes} that the descent morphism of the toposes is a local geometric surjection. Finally we show in section \ref{sec:ringdes} that the descent morphism of the rings is an epi in $\cat{Ring}^{op}([C(A(I\cup J)),\cat{Set}])$.

\subsubsection{Descent morphism for posets}\label{sec:posdes}
Theorem \ref{prop:main} relies on the fact that the locality of a net on $\field{R}^{1+1}$ is determined by the locality of the restrictions of the net to each spacelike line. 
\begin{lemma}\label{lem:restr}
Let $A: \mc{V}(\field{R}^2)\rightarrow \cat{CStar}_{inc}$ be a net. Then $A$ is strongly local iff its restriction to any spacelike axis is local. Moreover, $A$ is additive iff its restriction to any spacelike axis is additive.
\end{lemma}
\begin{proof}
If $A$ is strongly local, then clearly its restriction to a spacelike axis is strongly local and if $A$ is additive, then its restricition is.
Conversely, for any two spacelike separated causal diamonds $O_1$ and $O_2$, there is some axis $l$ and two spacelike separated diamonds $U_1$ and $U_2$ such that the spacelike diagonal of $U_1$ and $U_2$ lies on $l$ and that $O_1\subset U_1$ and $O_2\subset U_2$. The locality of the restriction of $A$ then implies that $A(U_1)$ and $A(U_2)$ commute in $A(U_1\vee U_2)$, so $A(O_1)$ and $A(O_2)$ commute as well. A similar argument applies to strong locality and additivity.
\end{proof}

Our aim will thus be to prove the properties of the descent morphism listed in theorem \ref{prop:main} for nets on the line. The crucial step in the proof is to note that the commutance of two algebras in a larger can be naturally expressed in terms of their posets of commutative subalgebras.

\begin{proposition}\label{prop:adj}
Let $A_1, A_2 \rInto A$ be two subalgebras of a C*-algebra $A$. Then $A_1$ and $A_2$ mutually commute in $A$ if and only if the functor
\begin{diagram}
C(A)		& \rTo^r 	& C(A_1)\times C(A_2)\\
C\subset A	& \mapsto	& (C\cap A_1, C\cap A_2)
\end{diagram}
has a left adjoint $\vee$.
\end{proposition}
\begin{proof}
Suppose that $A_1$ and $A_2$ indeed mutually commute in $A$. Then there is a functor $C(A_1)\times C(A_2)\rightarrow C(A)$, which is given by the map $(C_1, C_2)\mapsto C_1\vee C_2$ sending each pair of commutative subalgebras to the algebra in $A$ generated by them (i.e. the least subalgebra of $A$ containing both $C_1$ and $C_2$). Since the elements in $C_1$ commute with the elements in $C_2$, all generators of $C_1\vee C_2$ commute so that $C_1\vee C_2$ will be commutative. Moreover, the definition of $C_1\vee C_2$ states precisely that $C_1\vee C_2\leq C$ iff $(C_1,C_2)\leq r(C)$, so that $\vee\dashv r$.

Conversely, if $r$ has a left adjoint, then it must be given by $\vee$. In particular, we can take any normal element $c_1\in A_1$ and $c_2\in A_2$ and consider the commutative subalgebras $C_1\subset A_1$ and $C_2\subset A_2$ generated by them. Then $C_1\vee C_2$ is commutative, so $c_1$ and $c_2$ commute. We conclude that $A_1$ and $A_2$ mutually commute in $A$.
\end{proof}

We can directly apply this to \emph{any} local net.

\begin{corollary}\label{cor:local}
A net $A:\mc{V}(X)\rightarrow\cat{CStar}_{inc}$ is local if and only if one finds for any pair $(O_1,O_2)$ of spacelike separated opens that the functor
$$
C(A(O_1\vee O_2)) \rTo^r C(A(O_1))\times C(A(O_2))
$$
has a left adjoint $\vee$.
\end{corollary}

The strong locality of the net now precisely gives the following extension of this result:

\begin{lemma}\label{lem:inv}
A net is strongly local, precisely when the adjunction unit from proposition \ref{cor:local} is an isomorphism, or equivalently, when the left adjoint $\vee$ is full and faithful.
\end{lemma}
\begin{proof}
Indeed the definition of strong locality precisely states that $r\circ \vee$ is isomorphic to the identity. Since we are working with posets, this means that the adjunction unit should be an isomorphism.
\end{proof}

We will now restrict our attention to local nets on the line: for additive nets on the line, we can generalize this result to opens that are not spacelike separated. Note that a net $A: \mc{V}(\field{R})\rightarrow\cat{CStar}_{inc}$ on the open intervals of the line is additive if $A((a,c))=A((a,b))\vee A((b,c))$.

\begin{proposition}\label{prop:line}
Let $A:\mc{V}(\field{R})\rightarrow\cat{CStar}_{inc}$ be an additive net. Then $A$ is local if and only if for any interval $K$ which is covered by two intervals $I$ and $J$, we have that the functor
$$C(A(K))\rTo^r C(A(I))\times_{C(A(I\cap J))}C(A(J))$$
has a left adjoint $\vee$. Moreover, the net is strongly local iff the unit of the adjunction is an isomorphism, or equivalently, the left adjoint $\vee$ is full and faithful.
\end{proposition}
\begin{proof}
Note that the category $C(A(I))\times_{C(A(I\cap J))}C(A(J))$ has as objects pairs $(C_1,C_2)\in C(A(I))\times C(A(J))$ such that $C_1\cap A(I\cap J)=C_2\cap A(I\cap J)$. In particular, since $A(\emptyset)=\field{C}\cdot 1$, one finds that this simplifies to a direct product if $I$ and $J$ are disjoint intervals. In that case, we can apply corollary \ref{cor:local} to obtain the statement.

In case $I$ and $J$ are not disjoint, one can still let the left adjoint of $r$ send a pair $(C_1,C_2)$ of commutative subalgebras to the subalgebra $C_1\vee C_2$ generated by them in $A(K)$. Indeed, we can replace the two intervals $I$ and $J$ by the three disjoint intervals $I\cap J$, $I\setminus J$ and $J\setminus I$\footnote{We let $I\setminus J$ denote the interior of the set-theoretical complement}, as indicated in figure \ref{fig:2}. 
\begin{figure}[!t]
\centering
	\includegraphics[width=0.2 \textwidth,trim=50mm 90mm 50mm 70mm]{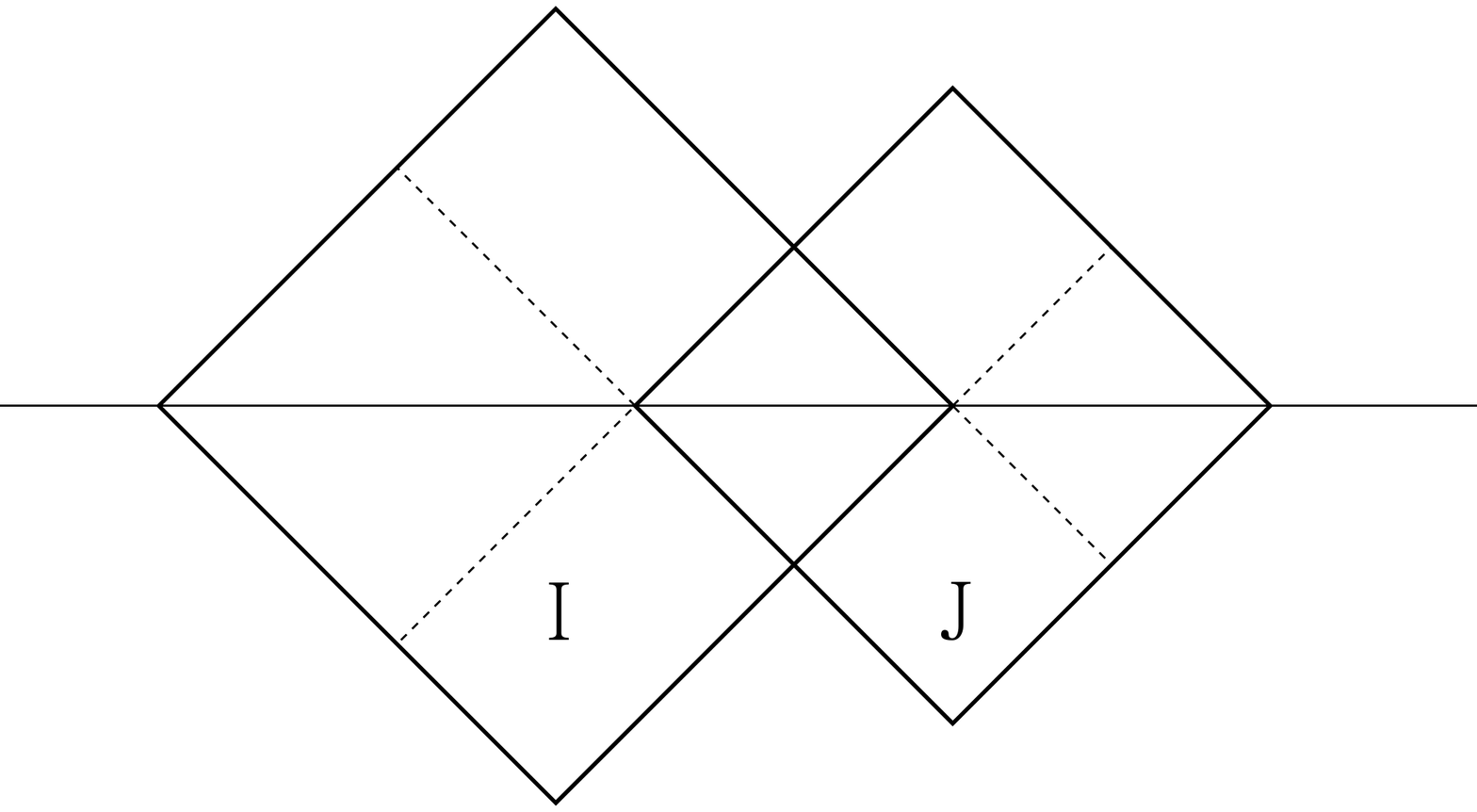}
\caption{\it \label{fig:2} The causal diamonds corresponding to $I\cap J$, $I\setminus J$ and $J\setminus I$.}
\end{figure}
One may note that the additivity property of the net $A$ implies that $A(I)=A(I\setminus J)\vee A(I\cap J)$ and similarly for $J$. In particular, this means that 
$$A(K)=A(I\setminus J)\vee A(I\cap J)\vee A(J\setminus I).$$ 
Moreover, these three algebras mutually commute and have trivial intersection in $A(K)$ since the corresponding intervals are disjoint (hence spacelike separated). 

Now, let $(C_1,C_2)\in C(A(I))\times C(A(J))$ be a pair of commutative subalgebras that agree on $A(I\cap J)$. We note that $C_1$ commutes with the part of $C_2$ in $A(I\cap J)$, since $C_1$ and $C_2$ are the same in $A(I\cap J)$. Furthermore, $C_1$ commutes with the part of $C_2$ in $A(J\setminus I)$, since $C_1\subset A(I)$ and $A(I)$ and $A(J\setminus I)$ mutually commute in $A(K)$. Because $C_2$ is a subalgebra of $A(J)=A(I\cap J)\vee A(J\setminus I)$, this means that the whole of $C_2$ has to commute with $C_1$.

Thus, $C_1$ and $C_2$ mutually commute in $A(K)$, so $C_1\vee C_2$ is a commutative subalgebra of $A(K)$. We thus find a well-defined map
$$C(A(I))\times_{C(A(I\cap J))}C(A(J))\rTo^{\vee}C(A(K))$$
which is easily seen to be a functor. The definition of strong locality gives that the unit of the adjunction is an isomorphism, precisely if the net is strongly local.
\end{proof}

If $r$ were an isomorphism, this result would state (part of) a sheaf property of the presheaf $I\mapsto C(A(I))$. We see that $C(A(-))$ is in fact a coseparated presheaf (see definition \ref{def:sh}) since $C(A(I\cup J))\rightarrow C(A(I))\times_{C(A(I\cap J))}C(A(J))$ is an epimorphism. Since $C(A)$ only serves as the site for a presheaf topos, we lift this result to the corresponding toposes in the next section.

\subsubsection{Descent morphism for toposes}\label{sec:topdes}
We now lift the result of proposition \ref{prop:line} to the corresponding presheaf toposes, using some category theoretic results from the appendix.

\begin{corollary}
Let $A:\mc{V}(\field{R})\rightarrow\cat{CStar}_{inc}$ be an additive net. Then $A$ is strongly local if and only if for every pair of intervals $(I,J)$ covering an interval $I\cup J$, the geometric morphism $r_*\vdash r^*$ extends to a quadruple of adjunctions
\begin{diagram}[LaTeXeqno]\label{diag:quad}
[C(A(I))\times_{C(A(I\cap J))}C(A(J)),\cat{Set}]	&
\pile{\lTo^{Lan_r} \\ \bot \\ \rInto^{r^*} \\ \bot \\ \lTo^{r_*} \\ \bot \\ \rInto_{\text{Ran}_\vee}}
& [C(A(I\cup J)),\cat{Set}].
\end{diagram}
in which both $r^*$ and $Ran_\vee$ are full and faithful.
\end{corollary}
\begin{proof}
$A$ is local iff we find from proposition \ref{prop:line} an adjunction
$$
[C(A(I))\times_{C(A(I\cap J))}C(A(J))]^{op}\pile{\lTo^{r^{op}}\\ \bot \\ \rTo_{\vee^{op}}}C(A(I\cup J))^{op}
$$
for every pair of intervals $(I,J)$. Lemma \ref{lem:adj} then states that one finds the quadruple of adjunctions in diagram \ref{diag:quad}, where $r_*=\vee^*$ and the further right adjoint is in fact $\vee_*$.
Lemmas \ref{lem:faith} and \ref{lem:ff} imply that $r^*$ and $Ran_\vee$ are full and faithful.

Conversely, given the adjoint quadruple, corollary \ref{cor:pshfunc} implies that $r_*$ is actually a precomposition functor $\vee^*$ for some functor 
$$\vee: C(A(I))\times_{C(A(I\cap J))}C(A(J))\rightarrow C(A(I\cup J))$$
since it is the inverse image of an essential geometric morphism.
We wish to show that $\vee$ is left adjoint to $r$. To do so, we note that by lemma \ref{lem:lanrep}, both $\text{Lan}_r$ and $\text{Lan}_\vee$ restrict to the observables. Denoting the Yoneda embedding by $Y$, we then have a commuting square
\begin{diagram}
[C(A(I))\times_{C(A(I\cap J))}C(A(J)),\cat{Set}]	&
\pile{\rInto^{\text{Lan}_\vee=r^*} \\ \bot \\ \lTo^{\text{Lan}_r}}	& [C(A(I\cup J)),\cat{Set}]\\
\uInto^Y	&	&	\uInto_Y\\
[C(A(I))\times_{C(A(I\cap J))}C(A(J))]^{op} &\pile{\lTo^{r^{op}}\\ \bot \\ \rTo_{\vee^{op}}} & C(A(I)\vee A(J))^{op}
\end{diagram}
so $\vee$ is indeed left adjoint to $r$. Since we find the adjunction from lemma \ref{prop:line}, we conclude that $A$ is local.
\end{proof}
As a direct corollary, we find that $r$ induces a local geometric morphism.
\begin{corollary}\label{cor:surj}
Let $A:\mc{V}(\field{R})\rightarrow\cat{CStar}_{inc}$ be an additive net. Then $A$ is strongly local iff for every pair of intervals $(I,J)$ covering an interval $I\cup J$, one has a local geometric surjection
\begin{diagram}
[C(A(I\cup J)),\Set ] &	& \pile{\lInto^{r^*}\\ \bot \\ \rTo_{r_*}\\ \bot \\ \lInto}& & [C(A(I)),\Set]\times_{[C(A(I\cap J)),\Set]} [C(A(J)),\cat{Set}]
\end{diagram}
where the pullback is computed in the subcategory $\cat{Top}$ of $\cat{Topos}$.
\end{corollary}
\begin{proof}
The statement then follows immediately if one considers the equivalence
$$
[C(A(I))\times_{C(A(I\cap J))}C(A(J)),\cat{Set}]\simeq [C(A(I)),\Set]\times_{[C(A(I\cap J)),\Set]}[C(A(J)),\cat{Set}]
$$
which follows from the fact that $[-,\Set]:\cat{Pos}\rightarrow\cat{Top}$ preserves all limits, according to corollary \ref{cor:pshfunc}.
\end{proof}

Combining this proposition with lemma \ref{lem:restr} at the beginning of section \ref{sec:posdes} implies proposition \ref{prop:main}: we find that an additive net is strongly local precisely when for any spacelike line in $\field{R}^2$ and two intervals $I$ and $J$ on it, one finds that 
\begin{diagram}
[C(A(I)\vee A(J)),\cat{Set}]&	\pile{\lInto^{r^*}\\ \bot \\ \rTo_{r_*}}& [C(A(I)),\Set]\times_{[C(A(I\cap J)),\Set]}[C(A(J)),\cat{Set}].
\end{diagram}
is a local geometric morphism. In particular, the geometric morphism $r_*\vdash r^*$ is a geometric surjection. Finally, we turn to the internal rings.

\subsubsection{Descent morphism for rings}\label{sec:ringdes}
To complete the description of the descent morphism, we consider the descent morphism of the corresponding internal rings.
The topos $[C(A(I))\times_{C(A(I\cap J))}C(A(J)),\cat{Set}]$ is now endowed with the pushout ring
$$
\ul{A}(I)\coprod_{\ul{A}(I\cap J)}\ul{A}(J);\quad (C_1,C_2)\mapsto C_1\otimes_{C_1\cap C_2}C_2
$$
where $C_1\otimes_{C_1\cap C_2}C_2$ is the tensor product over $C_1\cap C_2$, which is indeed the pushout in the category of commutative rings. The descent morphism now gives a ring homomorphism
$$
r^*\left( \ul{A}(I)\coprod_{\ul{A}(I\cap J)}\ul{A}(J) \right)\rTo^\eta \ul{A}(K)
$$
in the topos $[C(A(K)),\Set]$. For each commutative subalgebra $D\subset A(K)$, we find a ring homomorphism
$$
\eta_D: (D\cap A(I))\otimes_{D\cap A(I\cap J)}(D\cap A(J))\rightarrow D=\ul{A}(K)(D).
$$
If we assume that the net is Einstein causal, we then find that $\eta$ is an injection.

\begin{proposition}
For an additive and Einstein causal net $A$, we find that the descent morphism of the ring is monic.
\end{proposition}
\begin{proof}
Indeed, like in the proof of proposition \ref{prop:line}, we can replace the two intersecting intervals $I$ and $J$ by the three disjoint intervals $I\setminus J$, $I\cap J$ and $J\setminus I$. Then additivity and Einstein causality give that
$$
A(K)=A(I\setminus J)\otimes A(I\cap J) \otimes A(J\setminus I).
$$
In particular, we now have two commutative subrings  $(D\cap A(I))\subset A(I\setminus J)\otimes A(I\cap J)$ and $(D\cap A(J))\subset A(J\setminus I)\otimes A(I\cap J)$ that both restrict to the same ring in $A(I\cap J)$. Then it is obvious that their tensor product over the intersection with $A(I\cap J)$ embeds in $D$.
\end{proof}

We thus conclude that the descent morphism of the Bohrified consists of a geometric surjection and an injective ring homomorphism. This matches the intuition that the ring structure of a ringed topos usually behaves dually to the geometric structure.

\subsection{Application to boundary nets}
The proof of proposition \ref{prop:main} made essential use of restricting a net to a boundary: indeed, we described the locality of a net on $\field{R}^2$ by means of its restrictions to spacelike axes. In \cite{lr04}, the authors develope a general theory of AQFT with boundaries. In particular, they consider (conformal) nets on the upper half $\field{R}\times\field{R}_{>0}$ of Minkowski space. They show how to construct from a net on the boundary $\field{R}\times\{ 0\}$ a bulk net on the upper half, such that its restriction to the boundary returns the original boundary net. By doing so, bulk nets can be constructed that satisfy appropriate boundary conditions.

In particular, given a local net $A$ on the line, they define a \emph{trivial bulk net} $A^+$ by $A^+(O)=A(I)\vee A(J)$, where $I$ and $J$ are the two intervals on the boundary, obtained by considering the two null projections of $O$ on the boundary, as indicated in figure \ref{fig:3}. 
\begin{figure}[!t]
\centering
	\includegraphics[width=0.2 \textwidth,trim=50mm 140mm 50mm 80mm]{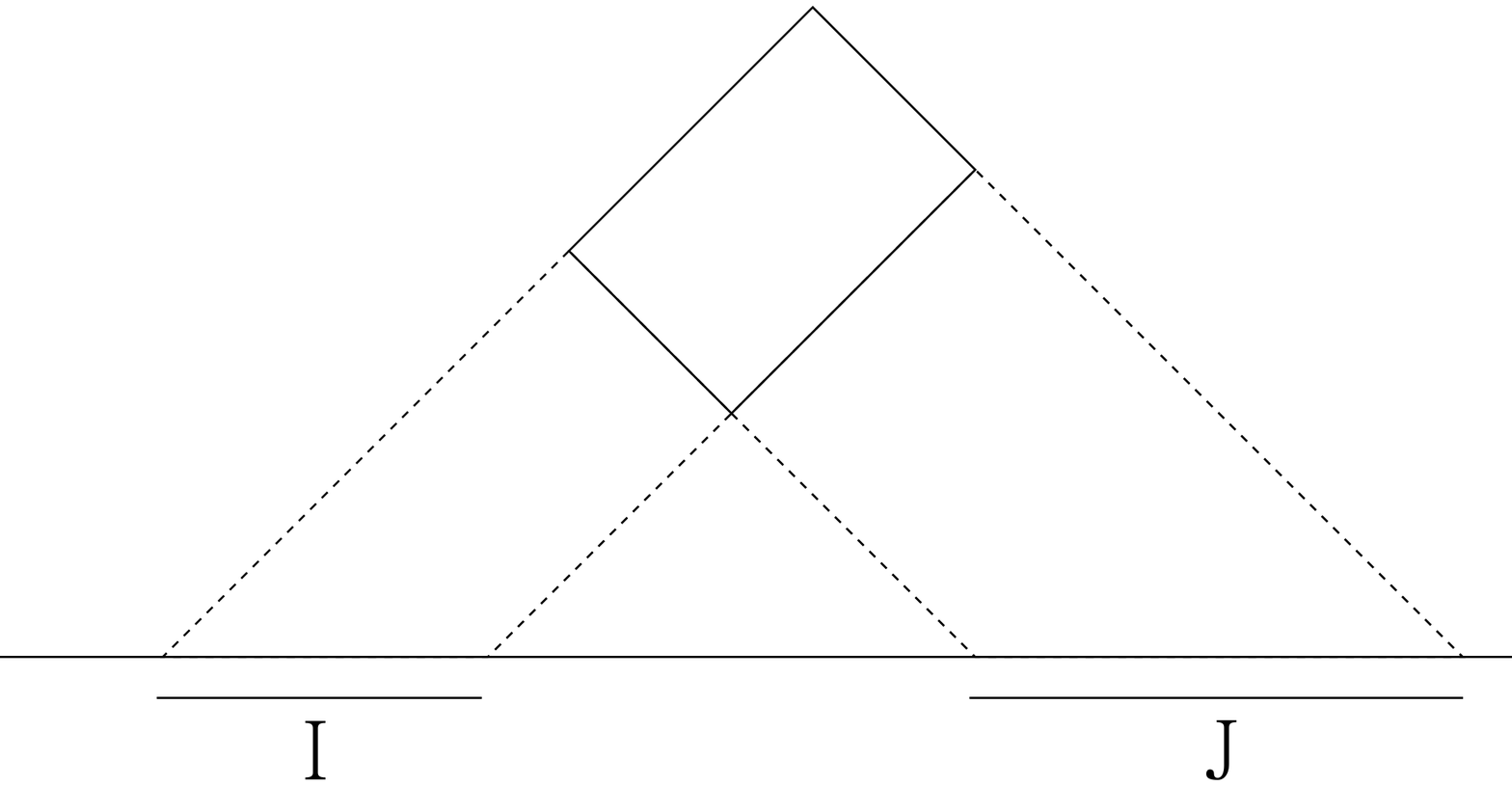}
\caption{\it \label{fig:3} The two null projections of $O$ on the boundary.}
\end{figure}
If $A$ is a local net, this indeed defines a local net on $\field{R}\times\field{R}_{>0}$. Denote the two null projections by $\pi_1$ and $\pi_2$ and let their equalizer be denoted by $\Delta: \field{R}\rInto \field{R}\times\field{R}_{>0}$. One then finds that $I=\pi_1(O)$, $J=\pi_2(O)$ and $I\cap J=\Delta^{-1}(O)$. One finds

\begin{corollary}
A boundary net $A$ is strongly local precisely if the following geometric morphism from the trivial bulk net $A^+$ is local for any open $O$ in the bulk:
\begin{diagram}
[C(A^+(O)),\Set ] &	& \pile{\lInto^{r^*}\\ \bot \\ \rTo_{r_*}\\ \bot \\ \lInto}& & [C(\pi_1^*A(O))\times_{C(\Delta_*A(O))}C(\pi_2^*A(O)),\cat{Set}].
\end{diagram}
\end{corollary}

\section{Outlook}
We have shown that Bohrification naturally extends from single quantum algebras to nets of algebras. The Bohrification construction shows that there is a strong relation between locality of nets with respect to spacetime on one hand, and locality with respect to the `classical contexts' on the other hand. In particular, we were able to characterize causal locality of a quantum field theory as its Bohrified presheaf of quantum phase spaces satisfying descent by a local geometric surjection, if the phase spaces are considered to be ringed Alexandrov spaces. 

The interplay between locality of nets on observables and locality of the Bohrified net in the sense of sheaf theory, suggests that Bohrification might lead to a sheaf theoretic description of algebraic quantum field theory. This would introduce the actual spacetime geometry in the picture of AQFT, thereby giving it a much more geometric flavour. This will be particularly important if one tries to consider local nets on curved spacetimes, as is done in \cite{bf09}.

The notion of a Bohrified net of ringed toposes described quantum field theory as quantum mechanics parametrized over spacetime. This can be made more explicit by applying the Grothendieck construction to the Bohrified net $B(A)$: one then obtains a single category $\int B(A)$ over the spacetime topos $\cat{Sh}(X)$ which contains all of the information about the QFT. In particular, this description contains all the local information, both with respect to the spacetime and the classical contexts. Since the Bohrified net assigned toposes to opens, the Grothendieck construction will give a category which is similar to the indexed toposes discussed extensively in \cite{joh02}. In should be investigated if this perspective might help to describe for instance the notion of a field. Fields might, for instance, have something to do with sections of the geometric morphism $\int B(A)\rightarrow \cat{Sh}(X)$. 

In any case, there are some points in quantum field theory where the Bohrification construction might add some insight. Our results therefore suggest that the ideas of Isham-D\"oring and Spitters et. al. might have some useful applications quantum field theory.

\appendix
\section{Category theoretic remarks}
In these sections, we collect some basic statements and definitions in category theory that we use in the main text.

\subsection{Adjunctions and Kan extensions}
We summarize some basic properties of adjunctions and the geometric morphisms they induce on presheaf categories. First recall that for any functor $f: \mc{C}\rightarrow\mc{D}$, there is a functor $f^*:[\mc{D}^{op},\Set]\rightarrow[\mc{C}^{op},\Set]$ that sends a presheaf $G: \mc{D}^{op}\rightarrow\Set$ to the presheaf $G\circ f^{op}: \mc{C}^{op}\rightarrow\Set$. A basic result in category theory now states that this functor $f^*$ has both a left and right adjoint, its left and right Kan extension, so that one finds an essential geometric morphism
$$
[\mc{C}^{op},\Set]\pile{\rTo^{{\rm Lan}_f} \\ \bot \\ \lTo^{f^*} \\ \bot \\ \rTo_{f_*={\rm Ran}_f}} [\mc{D}^{op},\Set].
$$
We note that $\text{Lan}_f$ is actually $f$ if it is restricted to the representable presheafs.
\begin{lemma}\label{lem:lanrep}
Let $f: \mc{C}\rightarrow\mc{D}$ be a functor. Then one finds for any representable $y_C=\text{hom}(-, C)$ in $[\mc{C}^{op},\Set]$ that
$$
\text{Lan}_fy_C=y_{f(C)}.
$$
\end{lemma}
\begin{proof}
Indeed, we find that
$$
\text{Nat}(\text{Lan}_fy_C, F)\simeq \text{Nat}(y_C, f^*F)\simeq F(f(C))\simeq \text{Nat}(y_{f(C)},F)
$$
so we conclude that $\text{Lan}_fy_C=y_{f(C)}.$
\end{proof}

In the case of an adjunction $r\vdash l$, it turns out that the geometric morphism induced by $l$ has a direct image $l_*={\rm Ran}_l$ that corresponds to the inverse image $r^*$ of the geometric morphism induced by $r$. 

\begin{lemma}\label{lem:adj}
Given an adjunction $\mc{D}\pile{\lTo^l \\ \bot \\ \rTo_r}\mc{C}$ one finds for the geometric morphisms $(l^*\dashv l_*):[\mc{C}^{op},\Set]\rightarrow [\mc{D}^{op},\Set]$ and $(r^*\dashv r_*):[\mc{D}^{op},\Set]\rightarrow [\mc{C}^{op},\Set]$ that $r^*=l_*$. In particular, one finds an adjoint triple
\begin{diagram}
[\mc{C}^{op},\Set]	& \pile{\lTo^{l^*} \\ \bot \\ \rTo^{r^*} \\ \bot \\ \lTo_{r_*}}	& [\mc{D}^{op},\Set].
\end{diagram}
\end{lemma}
\begin{proof}
By uniqueness of adjoints, it suffices to show that $l^*\dashv r^*$. The left adjoint of $r^*$ is the left Kan extension along $r$, which can be computed pointwise since the codomain of all presheaves is sufficiently nice. For a functor $F$ in $[\mc{D}^{op},\Set]$ one then finds
$$
(\text{Lan}_r F)(C)=\text{colim}_{C\rightarrow r(D)}F(D)=\text{colim}_{l(C)\rightarrow D}F(D)=(l^*F)(C)
$$
where the second identity follows from the adjunction $l\dashv r$. We conclude that $l^*$ agrees with the left adjoint of $r^*$ on the objects of $[D^{op},\Set]$. One can show that they are also the same on arrows.
\end{proof}

\begin{lemma}\label{lem:faith}
Let $f:\mc{C}\rightarrow\mc{D}$ be a full and faithful functor. Then the left Kan extension ${\rm Lan}_f:[\mc{C}^{op},\Set]\rightarrow [\mc{D}^{op},\Set]$ is full and faithful as well.
\end{lemma}
\begin{proof}
Note that for any adjunction $r\vdash l$, one has that the left adjoint $l$ is full and faithful precisely when the unit ${\rm id}\rightarrow rl$ is an isomorphism. In this case, we have an adjunction $f^*\vdash {\rm Lan}_f$, so we have to show that $f^* {\rm Lan}_f$ is isomorphic to the identity.

As in the previous lemma, we compute the left Kan extension pointwise. One then finds
$$
(f^* \text{Lan}_f F)(C)=\text{colim}_{f(C)\rightarrow f(D)}F(D)=F(C)
$$
where the latter equality follows from the fact that the comma category $f(C)/f$ has the object $(C,{\rm id}: f(C)\rightarrow f(C))$ as an initial object. We therefore see that $f^* {\rm Lan}_f$ is isomorphic to the identity.
\end{proof}

\begin{lemma}\label{lem:ff}
For any adjoint triple
\begin{diagram}
\mc{C}	& \pile{\rTo^f \\ \bot \\ \lTo^g \\ \bot \\ \rTo_h} & \mc{D}
\end{diagram}
one has that $f$ is full and faithful if and only if $h$ is.
\end{lemma}
\begin{proof}
Let ${\rm id}\rTo^{\eta_1} gf$ and ${\rm id}\rTo^{\eta_2} hg$ be the units of the two adjunctions.
Again, note that the left adjoint $f$ is full and faithful precisely when $\eta_1$ is an isomorphism. Now, composing the two adjunctions gives an adjunction $gh\vdash gf$. From the uniqueness of adjoints it now follows that when $gf\simeq \text{id}$, one also finds that $gh\simeq \text{id}$ and that the adjunction unit
$$
\text{id} \rTo^{\eta_1} gf \rTo{\eta_2 gf} ghgf
$$
is an isomorphism. But since $\eta_1$ is an isomorphism, this means that $\eta_2 gf$ is an isomorphism. Finally, since $gf\simeq \text{id}$, we find that the adjunction unit $\eta_2$ is an isomorphism, so $h$ is full and faithful as well.
\end{proof}

\subsection{Locales}\label{sec:loc}
A locale (\cite{joh02, mml92}) is a category theoretic generalization of a topological space. To motivate the definition, note that the topological information of a space $X$ is captured in its poset $\mc{O}(X)$ of opens, ordered by inclusion. The axioms for a topological space state precisely that $\mc{O}(X)$ has finite limits (intersections and a top element) and infinite colimits (unions and a bottom element). It is also clear that for any collection of opens $(U_i)_i$ and an open $V$, one finds that $V\cap(\bigcup_i U_i)=\bigcup_i(U_i\cap V)$. Thus, the topology of a space constitutes a highly structured kind of poset called a frame.

\begin{definition}
Let $P$ be a poset with infinite coproducts and finite products. For any two objects $a$ and $b$ in $P$, we denote their product (meet) by $a\wedge b$ and for any set $\{a_i\}_{i\in I}$ of objects we denote their coproduct (join) by $\bigvee_{i\in I}a_i$.

Then $P$ is called a frame if for all $a\in P$, the functor $-\wedge a: P\rightarrow P$ preserves all coproducts.
\end{definition}

A continuous map $f: X\rightarrow Y$ is charactrized by the fact that the inverse image restricts to a functor $f^{-1}:\mc{O}(Y)\rightarrow \mc{O}(X)$. It is easily seen that this functor preserves all finite products and all coproducts in $\mc{O}(Y)$.

\begin{definition}
A frame homomorphism is a functor between frames that preserves all finite products and all coproducts. Together with the frames, the frame homomorphisms give a category $\cat{Frm}$ of frames.

One defines the category of locales to be $\cat{Loc}=\cat{Frm}^{op}$. The objects of $\cat{Loc}$ are called locales. For a locale $L$, we denote the corresponding frame in $\cat{Frm}$ by $\mc{O}(L)$.
\end{definition}

The previous discussion makes clear that $\cat{Top}$ forms a subcategory of $\cat{Loc}$. However, not every frame can be realized as a frame of opens of some space, which is sometimes stated as a locale being a topological space with not enough points. The following proposition formalizes this.
\begin{proposition}
The inclusion of the category $\cat{Top}$ of topological spaces in $\cat{Loc}$ has a right adjoint
\begin{diagram}
\cat{Top}	& \pile{\lTo^{\;\text{Pt}\;}\\ \top \\ \rTo_{\iota}} &	 \cat{Loc}.
\end{diagram}
\end{proposition} 
A proof is given in section IX.3 of \cite{mml92}. The left adjoint assigns to each locale its space of points: for a topological space, a point would be a continuous map $*\rightarrow X$, so a frame map $\mc{O}(X)\rightarrow\{0,1\}=\Omega$. Thus, generalizing, we define $\text{Pt}(L)=\cat{Frm}(\mc{O}(L),\Omega)$, topologized by taking for each $V\in\mc{O}(L)$ the subset $\{p\in\text{Pt}(L)\; |\; p^{-1}(V)=1 \}$ to be an open. One can extend this to a functor and check that it is indeed left adjoint to the inclusion.

A frame naturally forms a site if we say that a family $(U_i\rightarrow U)_i$ covers $U$ if $\bigvee_i U_i=U$. Moreover, a frame homomorphism then constitutes a morphism of sites, so that we have a functor $\text{Sh}:\cat{Loc}\rightarrow\cat{Topos}$ to the category of (Grothendieck) toposes with geometric morphisms as arrows. We mention the following result:
\begin{proposition}\label{prop:locrefl}
The category $\cat{Loc}$ of locales is a reflective subcategory of $\cat{Topos}$, i.e. one has an adjunction
\begin{diagram}
\cat{Loc}	& \pile{\lTo^{\text{Sub}(1)}\\ \bot \\ \rTo_{\text{Sh}}}	& \cat{Topos}.
\end{diagram}
\end{proposition}
Indeed, one has for any Grothendieck topos that the subobjects of $1$ constitute a frame. Moreover, for the sheaves on a locale $L$, one finds that the subobjects of $1$ return precisely the elements of $\mc{O}(L)$.

Thus, toposes form a generalization of locales, which in turn generalize the notion of a topological space.

\subsection{Internal locales}\label{sec:inloc}
The theory of frames can be internalized to any topos: any frame is a (co)complete Heyting algebra, since we can define
$$
(a\Rightarrow b) = \bigvee_{a\wedge c\leq b}c.
$$
Thus, a frame is a Heyting algebra which is also a poset with infinite coproducts. Since a Heyting algebra structure is equationally defined (by maps $\wedge,\Rightarrow: L\times L\rightarrow L$ and $\top, \bot: 1\rightarrow L$), one can easily interpret it in any category with finite limits. For the second condition, we need the definition of a cocomplete poset (\cite{joh02}):
\begin{definition}
A poset in a topos $\mc{E}$ is an internal category $\field{P}$ such that the domain and codomain maps constitute a monic
$$
P_1\rInto^{(d_0,d_1)}P_0\times P_0
$$
and that one finds a pullback
\begin{diagram}
P_0			&	\rTo^{e}		& P_1\\
\dTo^{e}	&					& \dTo_{(d_0,d_1)}\\
P_1			& \rTo^{(d_1,d_0)}	& P_0\times P_0
\end{diagram}
where $e$ sends to each object the corresponding identity morphism.
\end{definition}
We see that this precisely states that there can be at most one arrow from one object to another and that arrows $a\rightarrow b$ and $b\rightarrow a$ imply both arrows to be identities. To state completeness, we need the notion of a `diagram in $\field{P}$'.

\begin{definition}
Let $I$ be an object in $\mc{E}$ and $\field{D}$ an internal category. Denote by $\mc{D}^I$ the (ordinary category) in which
\begin{itemize}
\item objects are arrows $I\rightarrow D_0$ in $\mc{E}$ and
\item an arrow from $I\rTo^f D_0$ to $I\rTo^g D_0$ is an arrow $I\rTo^{h}D_1$ so that $d_0\circ h = f$ and $d_1\circ h=g$.
\end{itemize}
The identity morphisms and composition are induced in a similar way from the structure of $\field{D}$. 
\end{definition}
We note furthermore that any morphism $J\rTo^{x}I$ induces a functor $\mc{D}^I\rTo^{x^*}\mc{D}^{J}$ by precomposition with $x$, so that we find a functor $\mc{E}^{op}\rTo \cat{Cat}$ (also called an $\mc{E}$-indexed category).

Thus, an object of $\mc{D}^I$ is a bunch of objects in the internal category $\field{D}$, indexed by the object $I$. In $\Set$, this is just a set of objects (and no arrows) in a category $\mc{D}$. In $\Set$, a category $\mc{D}$ has all coproducts of type $I$ if the functor $\Delta: \mc{D}\rightarrow [I,\mc{D}]$ has a left adjoint, so that $[I,\mc{D}](F,\Delta_C)\simeq\mc{D}(\colim F,C)$. Stated differently, the functor $t^*:[1,\mc{D}]\rightarrow[I,\mc{D}]$ induced by precomposition with the functor $1\rightarrow I$ should have a left adjoint. This description of a colimit can be formulated internally.

\begin{definition}\label{def:cocomp}
We say an internal poset $\field{P}$ is cocomplete (i.e. has all joins) if for any arrow $x: I\rightarrow J$ the corresponding pullback functor $x^*:\mc{P}^J\rightarrow \mc{P}^I$ has a left adjoint $\vee_x$, such that for any pullback diagram
\begin{diagram}\label{diag:BC}
H		& \rTo^a	& I			&	& \mc{P}^H			& \lTo^{a^*}& \mc{P}^I		\\
\dTo^b	&			& \dTo^c	&	& \dTo^{\vee_b}		&			& \dTo^{\vee_c}	\\
J		& \rTo^d	& K			&	& \mc{P}^L			& \lTo^{d^*}& \mc{P}^K
\end{diagram}
the second diagram commutes. 

Moreover, a functor $\field{P}\rTo\field{L}$ is said to preserve colimits iff the induced functor $F_*:\mc{P}^I\rightarrow\mc{L}^I$ commutes with taking colimits, in the sense that for any $x:I\rightarrow J$ we have that
\begin{diagram}
\mc{P}^J		&	\rTo^{F_*}	& \mc{L}^{J}\\
\uTo^{\vee_x}	&				& \uTo_{\vee_x}\\
\mc{P}^I		&	\rTo_{F_*}	& \mc{L}^I
\end{diagram}
commutes.
\end{definition}
\noindent One can define completeness in precisely the same way by replacing the left adjoint by a right adjoint. The inclusion of adjoints for any pullback functor $x^*$, not just the ones from $\mc{P}^1$, arises to guarantee that coproducts are `consistently computed'. The condition that the diagram \ref{diag:BC} should commute (called the Beck-Chevalley condition) is a mere technical condition which can be easily seen to hold in $\Set$.

Recall that the equational definition of a Heyting algebra $L$ via $\wedge$ and $\Rightarrow$ induces a poset-structure on $L$; indeed, one defines the category $\field{L}$ by taking $L_0=L$ and $L_1$ the equalizer
$$
L_1\rInto L_0\times L_0 \pile{\rTo^{\pi_1} \\ \rTo_{\wedge}} L_0
$$
(expressing $a\leq b$ iff $a\wedge b=a$). Thus, a frame is a Heyting algebra which, as a poset, is cocomplete. Moreover, the frames form a category $\cat{Frm}(\mc{E})$ with as arrows the (internal) functors preserving the meet and infinite coproducts. With this definition, one can show that frames are preserved by geometric morphisms.

\begin{lemma}\label{prop:geomfr}
The direct image $f_*$ of a geometric morphism $\mc{E}\rTo^f \mc{F}$ induces a functor $\cat{Frm}(\mc{E})\rTo^{f_*}\cat{Frm}(\mc{F})$.
\end{lemma}
\begin{proof}
Note that $f_*$ preserves limits, so it preserves Heyting algebras; that is, all defining diagrams containing the meet, Heyting implication and top and bottom elements are preserved.  Moreover, $f_*$ preserves cocomplete posets: indeed, for $\field{P}$ a poset in $\mc{E}$, $f_*\field{P}$ is a poset since this structure is equationally defined. Moreover, any functor
$x^*: f_*\mc{P}^J\rightarrow f_*\mc{P}^I$ corresponds via the adjunction $f^*\dashv f_*$ to $(f^*x)^*:\mc{P}^{f^*J}\rightarrow\mc{P}^{f^*I}$, of which we can take the left adjoint $\vee_{f^*x}$. Since $f_*$ preserves limits, the Beck-Chevalley condition is satisfied as well.

Moreover, $f_*$ preserves frame maps, since the conditions of a frame map are formulated in terms of commuting diagrams.
\end{proof}
It can be shown (cf. \cite{joh02}) that definition \ref{def:cocomp} is in fact equivalent to the following definition:
\begin{definition}
Let $\field{P}$ be an internal poset. Consider the characteristic map $P_0\times P_0\rTo\Omega$ of $P_1\rInto^{(d_0,d_1)}P_0\times P_0$. Then $\field{P}$ is cocomplete iff the transpose $P_0\rightarrow \Omega^{P_0}$ of this characteristic map has a left adjoint.
\end{definition}
In $\Set$ this makes sense: the characteristic map sends a pair $(x,y)$ in a poset $P$ to $1$ if $x\leq y$. The transpose of this map sends $y\in P$ to the subset $\{x\leq y\}\subset P$. Saying that $P$ has all coproducts then means that the desired left adjoint exists.

\begin{lemma}
In $\cat{Loc}(\mc{E})$, the point $*$ (having $\mc{O}(*)=\Omega$) is the terminal object.
\end{lemma}
\begin{proof}
For any frame $F$, there is a unique frame homomorphism $\Omega\rightarrow F$. It is given by the map $p \mapsto \bigvee \{\top_{F}| \top_\Omega = p\}$.
\end{proof}

Using lemma \ref{prop:geomfr} we can now prove the following theorem by Johnstone\cite{joh02}:
\begin{proposition}\label{prop:loc}
Let $X$ be a locale (internal to $\cat{Set}$). Then there is an equivalence of categories
$$
\cat{Loc}(\cat{Sh}(X))\simeq \cat{Loc}/X
$$
\end{proposition}
\begin{proof}
We first construct an internal locale from the bundle. Let $Y\rTo^\pi X$ be a continuous map, then by lemma \ref{prop:geomfr} it induces a functor $\cat{Loc}(\cat{Sh}(Y))\rTo^{\pi_*}\cat{Loc}(\cat{Sh}(X))$. In particular, one finds that $\pi_*(\Omega_Y)$ is a frame, given pointwise by
$$
\pi_*(\Omega_Y)(U)=\{V\in\mc{O}(Y)|V\leq\pi^*U \}.
$$

Moreover, one can extend this map to a functor by sending a commuting triangle
\begin{diagram}[LaTeXeqno]\label{diag:triang}
Y	&				& \rTo^h	&				&	Z	\\
	& \rdTo_{\pi_Y}	&			& \ldTo_{\pi_Z}	&		\\
	&				& X			&				&
\end{diagram}
to a frame homomorphism $\pi_{Z*}(\Omega_Z)\rightarrow \pi_{L*}(\Omega_L)$. Since $\pi_{Z*}$ preserves frame maps, it suffices to canonically define a frame map $\Omega_Z\rTo h_*(\Omega_Y)$. But $\Omega_Z$ is the initial frame, so there is only one such map. It is easily seen that this indeed induces the structure of a functor.
\\

For a functor $\cat{Loc}(\cat{Sh}(X))\rightarrow \cat{Loc}/X$, let $\mc{L}$ be an internal locale. For the total space of the corresponding bundle, we take $L=t_*(\mc{L})=\mc{L}(X)$, where $t_*$ is the direct image of the geometric morphism corresponding to the (unique) continuous map $X\rTo^t 1$. To construct the bundle map, we have to assign to each open in $X$ an open in $\mc{L}(X)$. 

To do so, we note that one has that $\mc{L}(U)=\text{Hom}(U,\mc{L})$ where $U$ is interpreted as a subterminal object. Then the restriction map $\mc{L}(X)\rightarrow \mc{L}(U)$ corresponds to precomposition of a natural transformation $X\rightarrow\mc{L}$ with the inclusion $U\rightarrow X$. By the adjoint functor theorem, the cocompleteness of $\text{Hom}(U,\mc{L})$ (following from the fact that $\mc{L}$ is a frame) implies the existence of a left adjoint to this restriction map. In particular, one finds a map $\sigma_U: \mc{L}(U)\rightarrow\mc{L}(X)=L$. We then take $\pi^{-1}(U)={\rm colim}\;\sigma_U$ as the definition of the map of locales $L\rightarrow X$. Moreover, an internal map $\mc{L}\rightarrow \mc{M}$ of locales induces a map $L\rightarrow M$ on the corresponding total spaces which gives a commuting triangle as in diagram \ref{diag:triang}.

An easy computation shows that these two functors are each others pseudo-inverse.
\end{proof}

A direct consequence of this proposition is that a continuous map $f: X\rightarrow Y$ induces an adjunction 
$$
\cat{Loc}(\cat{Sh}(X))\pile{\lTo^{f^\sharp} \\ \top \\ \rTo_{f_*}} \cat{Loc}(\cat{Sh}(Y)).
$$
Indeed, this adjunction corresponds to the adjunction $(\Sigma_f \dashv f^\sharp)$ obtained from a change of base is $\cat{Loc}$. Moreover, for an internal locale $\mc{L}$, applying the direct image functor $f_*$ gives the locale $f_*\mc{L}$ in $\cat{Sh}(Y)$, which corresponds to the bundle $L\rTo X \rTo^f Y$, which is precisely $\Sigma_f L$ where $\Sigma_f: \cat{Loc}/X \rightarrow \cat{Loc}/Y$ is the left adjoint in the change of base adjunction. In particular, the notation $f_*$ for the left adjoint $\cat{Loc}(\cat{Sh}(X))\rightarrow \cat{Loc}(\cat{Sh}(Y))$ is justified since it is actually the restriction of the direct image functor to the subcategories of locales. We note that $f^\sharp$ is \emph{not} the same as the inverse image functor $f^*$: the latter need not even preserve frames.

\subsection{Sheaves}\label{sec:sh}
The concept of a sheaf has already been mentioned in section \ref{sec:loc} and we will shortly discuss it here, following \cite{mml92}. A sheaf on a topological space $X$ is a way to describe a class of functions on $X$ which have a local character: for instance, being differentiable or continuous is a local property of a function $f$, in the sense that $f$ is differentiable iff its restriction to any open subset of $X$ is differentiable (if $X$ is in fact a $\mc{C}^1$-manifold). To any open subset $U\subset X$, we can assign the set of differentiable functions on $U$, which we will call $\mc{C}^1(U)$. Now, for $U_1\subset U_2$, there is an obvious map $\mc{C}^1(U_2)\rightarrow\mc{C}^1(U_1)$ that sends a function $f$ on $U_2$ to its restriction $f|_{U_1}$, which is still differentiable. Thus, $\mc{C}^1(-)$ defines a presheaf on the frame $\mc{O}(X)$ of opens of $X$. The presheaf $\mc{C}^1(-): \mc{O}(X)^{op}\rightarrow\Set$ now has some special properties that come from the local character of being differentiable:
\begin{itemize}
\item Suppose one has a cover $(U_i)$ of an open $U\subset X$ and consider two differentiable functions $f, g$ on $U$ such that $f|_{U_i}=g|_{U_i}$ for all $i$. Then obviously $f=g$.
\item Conversely, let $(U_i)$ cover $U$ and consistently define on every open $U_i$ a differentiable function $f_i$, so that $f_i|_{U_i\cap U_j}=f_j|_{U_i\cap U_j}$. Then we can glue all these functions together to obtain a differentiable function on $U$.
\end{itemize}

We can express these properties in terms of diagrams: for a cover $(U_i)$ of $U$, we can form the set of `matching families of differentiable functions', in the sense just described. This set is then given by the equalizer
\begin{diagram}
E & \rInto & \prod_i \mc{C}^1(U_i) & \pile{\rTo^p \\ \rTo_q} & \prod_{i,j} \mc{C}^1(U_i\cap U_j)
\end{diagram}
where $p$ maps a family $\{f_i \}$ to $\{ f_i|_{U_i\cap U_j} \}$ and $q$ maps a family $\{f_i \}$ to $\{f_j|_{U_i\cap U_j} \}$ (so for each intersection $U_i\cap U_j$, $p$ and $q$ pick the different functions that restrict to this open). For any $f\in\mc{C}^1(U)$ one has that the tuple $\{f|_{U_i} \}$ has the property that $f|_{U_i}|_{U_i\cap U_j}=f|_{U_i\cap U_j}=f|_{U_j}|_{U_i\cap U_j}$, so the product of the restriction maps $r: \mc{C}^1(U)\rightarrow \prod_i \mc{C}^1(U_i)$ has the property that $pr=qr$. We thus find a morphism $d: \mc{C}^1(U)\rightarrow E$ which is called the \emph{descent morphism}. The first property now states that the descent morphism is injective, while the second property says that the descent morphism is surjective.

We can generalize this picture to any presheaf $F$ by just replacing $\mc{C}^1$ by $F$ in the above diagram. We then call the equalizer $E$ the set of \emph{matching families} of the presheaf $F$. $F$ is said to be a separated presheaf if one finds for any open $U$ and a cover $(U_i)$ of $U$ that the descent morphism is an injection. We say that $F$ is a coseparated presheaf if the descent morphism $F(U)\rightarrow E$ is a surjection. A sheaf is both a separated and a coseparated presheaf:

\begin{definition}
Let $X$ be a topological space and $F\in [\mc{O}(X)^{op},\Set]$. We say that $F$ is a sheaf if for any open $U$ and a cover $(U_i)$ of $U$, one finds an  equalizer diagram
\begin{diagram}
F(U) & \rInto & \prod_i F(U_i) & \pile{\rTo^p \\ \rTo_q} & \prod_{i,j} F(U_i\cap U_j)
\end{diagram}
where $p$ maps a family $\{f_i \}$ to $\{ f_i|_{U_i\cap U_j} \}$ and $q$ maps a family $\{f_i \}$ to $\{f_j|_{U_i\cap U_j} \}$.
\end{definition}

We remark that this definition of a sheaf makes no reference to the particular properties of the category $\Set$. In particular, we can define a sheaf as follows:
\begin{definition}\label{def:sh}
Let $\mc{D}$ be a complete category and $F: \mc{O}(X)^{op}\rightarrow\mc{D}$ a functor. We can construct for any $U$ and a cover $(U_i)$ of $U$ a the object of matching families as the equalizer
$$
E \rInto \prod_i F(U_i) \pile{\rTo^p \\ \rTo_q} \prod_{i,j} F(U_i\cap U_j)
$$
where $p$ and $q$ are defined as above. The descent morphism is then the unique morphism $F(U)\rightarrow E$ that corresponds to all restrictions $F(U)\rightarrow F(U_i)$.

We call $F$ separated if this descent morphism is mono and cosperated if it is epi. $F$ is a sheaf when the descent morphism is an isomorphism.
\end{definition}

In particular, we will be interested in functors $\mc{O}(X)\rightarrow\cat{RingTopos}$ to the category of ringed toposes. According to the definition, we should check whether the category of ringed toposes is in fact complete, to see if a sheaf condition can be formulated for such functors. It turns out that the category $\cat{RingTopos}$ is indeed complete, as follows by the following theorem (see Lurie \cite{lu09} for a general version of this statement).

\begin{lemma}\label{lem:ringtopos}
The category of ringed toposes is complete. If $J\rightarrow\cat{RingTopos}$ is a diagram sending each $j\in J$ to the ringed topos $(\mc{E}_j,R_{\mc{E}_j})$, then its limit $(\mc{E},R_\mc{E})$ is computed as
\begin{itemize}
\item the limiting topos 
\begin{diagram}
\mc{E}:= \lim_{\leftarrow j} \mc{E}_j & \rTo^{\pi_{j*}\vdash \pi_j^*\;\;} & \mc{E}_j
\end{diagram}
of the diagram $J\rightarrow\cat{RingTopos}\rightarrow\cat{Topos}$.
\item with as its internal ring the colimiting ring of the inverse image rings
$$
\pi_j^*R_{\mc{E}_j}\rightarrow\lim_{\rightarrow j}\pi_j^*R_{\mc{E}_j}=:R_{\mc{E}}
$$
\end{itemize} 
\end{lemma}
\begin{proof}

Suppose one has a diagram $J\rightarrow\cat{RingTopos}$ which is described by objects $(\mc{E}_j, R_j)$ and morphisms consisting of a geometric morphism $\mc{E}_i\rTo^{h_{ij}}\mc{E}_j$ and a ring homomorphism $h_{ij}^*R_{\mc{E}_j}\rTo^{\eta_{ij}}R_{\mc{E}_i}$.
The universal property is now easily checked: a cone $(\mc{F},R_{\mc{F}})\rTo(\mc{E}_j,R_{\mc{E}_j})$ consists of
\begin{itemize}
\item geometric morphisms $\mc{F}\rTo^{f_i}\mc{E}_i$ so that there is an equivalence $h_{ij}f_i\simeq f_j$ for all $i$ and $j$
\item ring homomorphisms $f_i^*R_{\mc{E}_i}\rTo^{\epsilon_i}R_\mc{F}$ so that for all $i$ and $j$ one finds a triangle 
\begin{diagram}
f_j^*R_{\mc{E}_j} & \rTo^\simeq & f_i^*h_{ij}^*R_{\mc{E}_j} & \rTo^{f_i^*\eta_{ij}}& f_i^*R_{\mc{E}_i}\\
& \rdTo(4,2)_{\epsilon_j} & &  & \dTo_{\epsilon_i}\\
& & & & R_\mc{F}
\end{diagram}
that commutes up to isomorphism, so that internally, one finds a cocone of rings.
\end{itemize}
The geometric morphisms $\mc{F}\rTo^{\pi_i}\mc{E}_i$ then factors over the limit in $\cat{Topos}$ and we find a unique geometric morphism
$$
\mc{F}\rTo^{u} \lim_{\leftarrow j} \mc{E}_j =: \mc{E}.
$$
The cocone of rings in $\mc{F}$ then takes the form 
$$
u^*\pi_i^*R_{\mc{E}_i}\rTo^{\epsilon_i}R_\mc{F}$$
where $\pi_i$ is the geometric morphism $\mc{E}\rightarrow \mc{E}_i$. This means that in $\mc{E}$, one finds a cocone $\pi_i^*R_{\mc{E}_i}\rTo u_* R_\mc{F}$ that factors uniquely over the colimit $R_\mc{E} = \colim \pi_i^*R_{\mc{E}_i}$. Transposing back, we then find the unique ring homomorphism $u^*R_\mc{E}\rightarrow R_{\mc{F}}$.

We thus conclude that the limit of $J\rightarrow\cat{RingTopos}$ is indeed as described in the proposition.
\end{proof}

A similar argument shows that the category $RingSp$ of ringed topological spaces is complete. It thus makes sense to talk about sheaves of ringed spaces or sheaves of ringed toposes.

\subsection{Constructive Gelfand Duality}\label{sec:gelf}
The classical Gelfand duality establishes a duality between the compact Hausdorff topological spaces and the commutative (unital) C*-algebras via the isomorphism $C\simeq\mc{C}(\Sigma_C,\field{C})$. One can give a constructive proof of this fact, which can be internalized to any (Grothendieck) topos.

\begin{theorem}[\cite{bm06, cs082}]
In any sheaf topos $\mc{E}$ one has an equivalence
$$
\cat{cCStar}(\mc{E})\simeq \cat{KRegLoc}(\mc{E})^{op}
$$
\end{theorem}

An internal *-algebra $A$ is a vector space over the complex rationals\footnote{Every Grothendieck topos has a natural numbers object, from which the rationals and the complex rationals can be derived.} $\field{C}_{\field{Q}}=\field{Q}+i\field{Q}$ with a multiplication $\cdot: A\times A\rightarrow A$ and an involution $*:A\rightarrow A$, such that the defining diagrams commute.

Furthermore, we endow $A$ with a norm, which is given by a relation $N\rInto A\times\field{Q}_+$ which has the characteristic properties of the classical norm relation $(a,q)\in N$ iff $||a||\leq q$ and finally, one requires $A$ to be complete (in some suitable sense). In particular, we note that

\begin{proposition}\label{prop:cstar}
Any presheaf $C^{op}\rightarrow\cat{CStar}$ is an internal C*-algebra in the category $[C^{op},\Set]$.
\end{proposition}
\begin{proof} Essentially Thm. 5 in \cite{hls09}.

\end{proof}

The construction of the spectrum of such an algebra is quite similar to the classical case. Indeed, the opens of the classical spectrum are generated by the basic opens $D_a=\left\{\phi:A\rightarrow\field{C}|\phi(a)\geq 0\right\}$ for any $a\in A_{sa}$ and therefore satisfy
\begin{align*}
D_1 &= \top\\
D_a\wedge D_{-a}&=\bot\\
D_{-b^2}&=\bot\\
D_{a+b}&\leq D_a\vee D_b\\
D_{ab}&=(D_a\wedge D_b)\vee(D_{-a}\wedge D_{-b})
\end{align*}
together with
$$
D_a\leq\bigvee_{r\in\field{Q}^+}D_{a-r}.
$$
Conversely, the topology is in fact defined by these relations. It turns out (cf. \cite{bm06, cs082}) that one can prove constructively that this gives the spectrum $\Sigma$, so that this description holds in any topos. Any commutative C*-algebra $A\in\mc{E}_0$ is then *-isomorphic to the commutative C*-algebra of continuous functions on the spectrum $\mc{C}(\Sigma,\field{C})$. In particular, we note that a point of the spectrum (so a continuous map $*\rightarrow\Sigma$) corresponds precisely to a *-homomorphism $A\simeq\mc{C}(\Sigma,\field{C})\rightarrow \mc{C}(*,\field{C})={\rm pt}(\field{C})$. Thus, as in the classical case, we conclude that
\begin{lemma}\label{lem:point}
The points of the spectrum $\Sigma$ of a commutative C*-algebra $A$ correspond to the *-homomorphisms $A\rightarrow\field{C}$, where $\field{C}$ is the C*-algebra of complex numbers. 
\end{lemma}

Finally, we will give the computation of the spectrum, as summarized in \cite{hls09}.
Given a C*-algebra $A$, let $A^+=\{a\in A_{sa}| a\geq 0 \}$ be the subobject consisting of the positive, self-adjoint elements of $A$. Recall that $A_{sa}$ has a natural ordering defined by $a\leq b$ if $b-a = c*c$ for some $c\in A$. One then defines a equivalence relation on $A^+$ by letting $a\approx b$ if there are $m,n\in\field{N}$ so that $a\leq mb$ and $b\leq na$. Then $L_A=A^+/_\approx$ is a lattice under the ordering induced by $\leq$.

One might notice that $L_A$ is a lattice that satisfies the first 4 conditions of the spectrum, if we interpret $D_a$ as the congruence class of $a\vee 0$ (which is always positive). We now have the following result:

\begin{lemma}[\cite{hls09}]
Let $\mc{E}$ be a Grothendieck topos and $A$ an internal C*-algebra. We define a covering relation on $L_A$ by stating that for every $D_a\in L_A$ and $U\in\Omega^{L_A}$ we have that $D_a\lhd U$ if and only if for all $q\in\field{Q}^+$ there is a finite $U_0\subset U$ such that $D_{a-q}\leq\bigvee U_0$.

Then the spectrum of $A$ is given by the frame $\{U\in\Omega^{L_A} | D_a\lhd U\Rightarrow D_a\in U \}$.
\end{lemma}
\noindent
The spectrum of $A$ can now be computed by interpreting all expressions in the topos in question.

\end{document}